\documentclass[journal]{IEEEtran}

\usepackage{amsmath} 
\usepackage{bm}
\usepackage{bbm}
\usepackage{amssymb}  
\usepackage{tikz}
\usepackage{caption,subcaption}
\usepackage{graphicx}
\usepackage{float}
\usepackage{cite}

\usepackage[boxed]{algorithm2e}
\SetAlCapSkip{1em}

\usepackage{amsthm}

\usepackage{paralist}

\usepackage[hidelinks]{hyperref} 

\newtheorem{theorem}{Theorem}

\newtheorem{lemma}[theorem]{Lemma}

\usepackage{pgfplots}
\usepackage{pgfplotstable}
\usepgfplotslibrary{groupplots}
\usepgfplotslibrary{fillbetween}

\pgfplotsset{
    compat=1.8,
    scale only axis,
    every tick label/.append style={font=\footnotesize},
    every axis/.append style={
        xlabel shift=0.5em,
        ylabel shift=0.5em
    }
}

\newlength{\figureheight}
\newlength{\figurewidth}

\newcommand{\n}{\widehat{n}}
\newcommand{\E}[1]{\mathbb{E}\left[#1\right]}

\newcommand{\Son}{\mathcal{S}^{\text{on}}}
\newcommand{\Soff}{\mathcal{S}^{\text{off}}}

\newcommand{\cY}{\mathcal{Y}}
\newcommand{\pd}[2]{p\left(#1\mid#2\right)}

\def\approxprop{%
  \def\p{%
    \setbox0=\vbox{\hbox{$\propto$}}%
    \ht0=0.6ex \box0 }%
  \def\s{%
    \vbox{\hbox{$\sim$}}%
  }%
  \mathrel{\raisebox{0.7ex}{%
      \mbox{$\underset{\s}{\p}$}%
    }}%
}

\makeatother

\begin{document}

\title{Improved event-based particle filtering in resource-constrained remote state estimation}

\author{Johan~Ruuskanen
        and~Anton~Cervin
\thanks{J. Ruuskanen and A. Cervin are with the Department of Automatic Control at Lund Univeristy, Sweden.}
\thanks{E-mail: \texttt{\{johan.ruuskanen, anton.cervin\}@control.lth.se}}}

\maketitle

\begin{abstract}
Event-based sampling has been proposed as a general technique for lowering the average communication rate in remote state estimation, which can be important in scenarios with constraints on resources such as network bandwidth or sensor energy. Recently, the interest of applying particle filters to event-based state estimation has seen an upswing, partly to tackle nonlinear and non-Gaussian problems, but also since event-based sampling does not allow an analytic solution for linear--Gaussian systems. Thus far, very little has been mentioned regarding the practical issues that arise when applying particle filtering to event-based state estimation. In this paper, we provide such a discussion by (i) demonstrating that there exists a high risk of sample degeneracy at new events, for which the auxiliary particle filter provides an intuitive solution, (ii) introducing a new alternative to the local predictor approach based on precomputing state estimates which is better suited to solve the issue of observer-to-sensor communication for closed-loop triggering in difficult systems, and (iii) exploring the difficulties surrounding the increased computational load when implementing the particle filter under event-based sampling.

\end{abstract}

\IEEEpeerreviewmaketitle

\section{Introduction} \label{sec:introduction}

Remote state estimation is ubiquitous in sensor-rich applications, where numerous IoT devices produce large amounts of measurements that need to be filtered and aggregated at, e.g., the nearest base station \cite{Boubiche18}. The setup is typically asymmetrical, in the sense that wireless sensor devices have limited computing and energy resources, while the base station has powerful CPUs and corded energy access. As a general technique to limit the network load and to prolong the lifetime of battery-operated devices, event-triggered communication of sensor data has been proposed in several works, e.g., \cite{AstBer99,Muller00,PerGus01, Kolios2016, Santos2019}. While appealing from an engineering viewpoint, event-based sampling also leads to more challenging state estimation problems, since the absence of a measurement also contains some information that should be used by the observer.

The sensor-side decision of when to trigger and send measurements to the remote estimator is in general an open problem, and the choice of triggering rule is highly dependent on the system dynamics and 
the performance objectives \cite{Sijs14}, \cite{Trimpe15}. The various types of triggering rules can be roughly split into two groups, namely \emph{open-loop triggers}, where the sensor independently decides when to trigger, and \emph{closed-loop triggers}, where the triggering rule is dependent on the state estimate in the observer \cite{Mohammadi17}. Two classical examples from these groups are \emph{send-on-delta} (SOD), where triggering is based on the previously sent measurement \cite{Miskowicz06}, and \emph{innovation-based triggering} (IBT), where triggering is based on the observer innovation \cite{Wu13}. An illustration of the closed-loop triggering strategy is given in Figure~\ref{fig:CL_eventtrigger}. Considering the error/communication trade-off, it is preferable to use a closed-loop trigger \cite{Trimpe15, Ruuskanen20}, but to actually obtain the observer information at the sensor is a potential problem. Traditionally, it has either been solved by letting only the sensor-to-observer communication be event based \cite{Wu13, Li19}, or by running a local filter or predictor at the sensor \cite{Shi14, Kolios2016, Li17, Liu18, SenLi19, Santos2019}. 

\begin{figure}
\medskip
	\centering
	\tikzset{every picture/.style={line width=0.75pt}} 

\begin{tikzpicture}[x=0.75pt,y=0.75pt,yscale=-0.8,xscale=0.8]

\draw    (270,117) -- (283.5,99) ;
\draw [shift={(283.5,99)}, rotate = 306.87] [color={rgb, 255:red, 0; green, 0; blue, 0 }  ][fill={rgb, 255:red, 0; green, 0; blue, 0 }  ][line width=0.75]      (0, 0) circle [x radius= 3.35, y radius= 3.35]   ;

\draw    (221.5,118) -- (270,117) ;
\draw [shift={(270,117)}, rotate = 358.82] [color={rgb, 255:red, 0; green, 0; blue, 0 }  ][fill={rgb, 255:red, 0; green, 0; blue, 0 }  ][line width=0.75]      (0, 0) circle [x radius= 3.35, y radius= 3.35]   ;

\draw  [dash pattern={on 0.84pt off 2.51pt}]  (295.85,117) -- (411.5,117) ;
\draw [shift={(413.5,117)}, rotate = 180] [color={rgb, 255:red, 0; green, 0; blue, 0 }  ][line width=0.75]    (10.93,-3.29) .. controls (6.95,-1.4) and (3.31,-0.3) .. (0,0) .. controls (3.31,0.3) and (6.95,1.4) .. (10.93,3.29)   ;
\draw [shift={(293.5,117)}, rotate = 0] [color={rgb, 255:red, 0; green, 0; blue, 0 }  ][line width=0.75]      (0, 0) circle [x radius= 3.35, y radius= 3.35]   ;
\draw   (258,99.8) .. controls (258,94.39) and (262.39,90) .. (267.8,90) -- (297.2,90) .. controls (302.61,90) and (307,94.39) .. (307,99.8) -- (307,132.2) .. controls (307,137.61) and (302.61,142) .. (297.2,142) -- (267.8,142) .. controls (262.39,142) and (258,137.61) .. (258,132.2) -- cycle ;
\draw    (449,138) -- (449,177) -- (284,176) -- (284,144) ;
\draw [shift={(284,142)}, rotate = 450] [color={rgb, 255:red, 0; green, 0; blue, 0 }  ][line width=0.75]    (10.93,-3.29) .. controls (6.95,-1.4) and (3.31,-0.3) .. (0,0) .. controls (3.31,0.3) and (6.95,1.4) .. (10.93,3.29)   ;

\draw  [fill={rgb, 255:red, 109; green, 165; blue, 233 }  ,fill opacity=0.8 ] (152,105) .. controls (152,100.58) and (155.58,97) .. (160,97) -- (214,97) .. controls (218.42,97) and (222,100.58) .. (222,105) -- (222,129) .. controls (222,133.42) and (218.42,137) .. (214,137) -- (160,137) .. controls (155.58,137) and (152,133.42) .. (152,129) -- cycle ;
\draw  [fill={rgb, 255:red, 236; green, 184; blue, 95 }  ,fill opacity=0.8 ] (414,105) .. controls (414,100.58) and (417.58,97) .. (422,97) -- (476,97) .. controls (480.42,97) and (484,100.58) .. (484,105) -- (484,129) .. controls (484,133.42) and (480.42,137) .. (476,137) -- (422,137) .. controls (417.58,137) and (414,133.42) .. (414,129) -- cycle ;

\draw (282,128) node [scale=0.8] [align=left] {\textbf{{\small Trigger}}};
\draw (360,101) node [scale=0.8]  {$y_{k}$};
\draw (359,160) node [scale=0.8]  {$H_k$};
\draw (187,117) node [scale=0.8] [align=left] {\textbf{Sensor}};
\draw (449,117) node [scale=0.8] [align=left] {\textbf{Observer}};

\end{tikzpicture}
	\smallskip
	\caption{Remote event-based state estimation with a closed-loop trigger. The event trigger depends on information from the observer, which must be transferred back from the observer to the sensor.}
	\label{fig:CL_eventtrigger}
\end{figure}
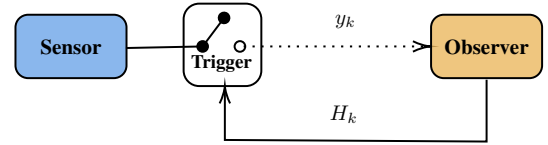

Recently, the use of particle filters \cite{Doucet08} for event-based state estimation has seen an upswing in attention from researchers \cite{Cea12, Davar17, Li19}. Particle filters provide a general way of performing state estimation for nonlinear and/or non-Gaussian systems where analytical solutions are intractable. Event-based sampling is by nature nonlinear, and it is in general impossible to find analytic solutions, even for linear--Gaussian (LG) systems \cite{Shi16}.

Previous work concerning particle filtering under event-based sampling has mainly been concerned with how to apply particle filtering to event-based state estimation under different triggering rules and system settings. However, it is widely known that particle filters suffer from drawbacks in the form of long computation times, lack of convergence guarantees, sample degeneracy, and poor performance in high dimensions. How limitations such as these affect the usability of particle filters in event-based systems remains a largely unexplored area. The aim of this paper, which is an extension of our two conference papers \cite{Ruuskanen19,Ruuskanen20}, is to shed some light on these issues and to present improvements to increase the applicability of particle filters under event-based sampling.

In particular, we discuss three problems arising from applying event-based sampling to particle filtering, how they relate to known drawbacks of particle filters, and how they can be mitigated:
\begin{enumerate}
    \item At new event instances, the risk of particle degeneracy becomes high. We show that the auxiliary particle filter constitutes an intuitive mitigation strategy, and derive a general way of approximating the fully-adapted particle filter for nonlinear event-based systems. 
    \item The problem of observer-to-sensor communication to achieve closed-loop triggering is exacerbated for particle filters. We provide a novel mitigation strategy based on precomputing state estimates, and present theoretical results on the optimal precomputation horizon length.
    \item The computational load can be large, especially for systems where measurement dimensionality is greater than one. We show that the approach of acceptance/rejection sampling will generate quick and adequate results for most well-behaved systems, and discuss how the event-based particle filter can be efficiently implemented. 
\end{enumerate}
Finally, a simulation study is performed to demonstrate the findings. As with most material so far published on event-based particle filters, this work focuses on the single sensor/single observer case.

\subsection{Related work}


Some notable works regarding event-based state estimation are presented in this section. Early work included contributions such as \cite{Sinopoli04}, where the information between events were simply omitted in a Kalman filter. Later, in \cite{Suh07} a Kalman filter using send-on-delta that integrated the nonmeasurement information to increase estimation quality was presented. Further in \cite{Wu13}, IBT was introduced and an approximate MMSE derived from the Kalman filter for an event-based LG system. It was later shown that it is possible to derive the exact MMSE for an LG system under a particular stochastic triggering scheme \cite{Han15}. Further theoretical foundation for event-based state estimation under a common event trigger was presented in \cite{Sijs14}.

Concerning nonlinear filtering and particle filters under event-based sampling, some early work exploring the notion was performed in \cite{Cea12}. Later, the application of the particle filter to SOD sampling was performed in \cite{Sid16, Davar17}, and to stochastic sampling in \cite{Sadeghzadeh16}. In \cite{Chitraganti17}, the particle filter was later applied to sampling with IBT and comparisons of the resulting estimator to the Cram\'er--Rao bound were made. Furthermore, in \cite{Li19}, the authors developed a trigger for non-Gaussian noise distributions to be used in conjunction with a particle filter. A hybrid event-sampling scheme using fusion with quantized data to improve performance was realized using the particle filter in \cite{Mohammadi18}. Finally, the application of event-based particle filters to estimation problems in power systems and smart grids was explored in \cite{Liu18}, \cite{SenLi19}. 

When the IBT was introduced \cite{Wu13}, it assumed periodic observer-to-sensor communication as a way to transfer information back to the sensor. As shown in \cite{Shi14}, the same performance could be obtained with only a simple sensor-side predictor and observer-to-sensor communication at new events. In contrary, a complete local filter solution for an LG system is shown in \cite{Li17}. The closed-loop triggering problem was also briefly discussed in \cite{Mohammadi17}. For particle filters with closed-loop triggering, only periodic observer-to-sensor communication has so far been considered.

\subsection{Outline}

In Section \ref{sec:background}, we for completeness shortly present the theory of event-based state estimation and particle filters to give the necessary background to understand the rest of the paper. In Sections \ref{sec:event_sample_degeneracy}, \ref{sec:closedloop_trouble} and \ref{sec:comp_aspects}, we then in turn discuss the three issues as presented in the contributions. Further, in Section \ref{sec:evaluation} we provide a simulation study to validate the proposed improvements. Finally, in Sections \ref{sec:discussion} and \ref{sec:conclusion}, we discuss our findings and conclude the paper.

\section{Background}  \label{sec:background}

\begin{figure*}[!t]
    \small	\centering
	\begin{subfigure}[t]{0.49\linewidth}
    	\centering
	    \label{fig:trigexample_m}
	    \vspace{-1ex}
    	\begin{tikzpicture}[]

\begin{axis}[%
width=0.9\columnwidth,
height=0.35\columnwidth,
scale only axis,
separate axis lines,
xmin=1,
xmax=25,
ymin=-7,
ymax=25,
ylabel={$y_k$},
xlabel={Time step $k$},
ylabel shift = -0.25em,
legend style={draw, fill=white, font=\footnotesize, at={(0.0,1.0)}, anchor=north west},
legend columns=3,
legend image post style={scale=0.5},
]

\def \filename {graphics/data/triggerexample_measurements.csv}
\def \filenametrig {graphics/data/triggerexample_trigs.csv}

\addplot [color=blue, thick] table [x=X, y=Y, col sep=comma] {\filename};
\addlegendentry{$y_k$};

\addplot [const plot, name path=upper, draw=none, forget plot] table [x=Xl, y=HU, col sep=comma] {\filename};
\addplot [const plot, name path=lower, draw=none, forget plot] table [x=Xl, y=HL, col sep=comma] {\filename};
\addplot [color=green!75!blue, fill opacity=0.3, area legend] 
    fill between [of = lower and upper]; 
\addlegendentry{$H_k$};

\addplot [const plot, ultra thick, opacity=0.3] table [x=X, y=Y, col sep=comma] {\filenametrig};
\addlegendentry{$\gamma_k = 1$}

\end{axis}

\end{tikzpicture}%
    	\caption{Event-sampling using IBT.}
	\end{subfigure}
    \begin{subfigure}[t]{0.49\linewidth}
    	\centering
	    \label{fig:trigexample_p}
	    \vspace{-1ex}
    	\begin{tikzpicture}[]

\begin{axis}[%
width=0.9\columnwidth,
height=0.35\columnwidth,
scale only axis,
separate axis lines,
clip marker paths=true,
xmin=1,
xmax=25,
ymin=-25,
ymax=30,
ylabel={$x_k$},
xlabel={Time step $k$},
ylabel shift = -1em,
legend style={draw, fill=white, font=\footnotesize, at={(0.0,1.0)}, anchor=north west},
legend columns=3,
legend image post style={scale=0.5},
]

\def \filename {graphics/data/triggerexample_measurements.csv}
\def \filenametrig {graphics/data/triggerexample_trigs.csv}
\def \filenamepart {graphics/data/triggerexample_part.csv}
\def \filenameline {graphics/data/triggerexample_part_lines.csv}


\addplot [color=black, thin, forget plot] table [x=X, y=Y] {graphics/data/triggerexample_part_lines2.txt};


\addplot [blue!50, mark=*, mark options={solid,scale=1, draw=black}, only marks] table [x=X, y=Xp, col sep=comma] {\filenamepart};
\addlegendentry{$X^i_k$};

\addplot [color=red, thick] table [x=X, y=Xh, col sep=comma] {\filename};
\addlegendentry{$x_k$};

\addplot [const plot, ultra thick, opacity=0.3] table [x=X, y=Y, col sep=comma] {\filenametrig};
\addlegendentry{$\gamma_k = 1$}

\end{axis}

\end{tikzpicture}%
    	\caption{Particles in the resulting state estimation.}
	\end{subfigure}
	\caption{Event-based state estimation with IBT and using the particle filter on the nonlinear, multimodal example system \eqref{eq:nonlin_example_system}. When $y_k$ steps outside the set $H_k$, an event $\gamma_k=1$ is triggered, and the resulting measurement sent to the observer.}
	\label{fig:trigexample}
	\par
	\noindent\makebox[\linewidth]{\rule{\linewidth}{0.4pt}}\par\smallskip
\end{figure*}
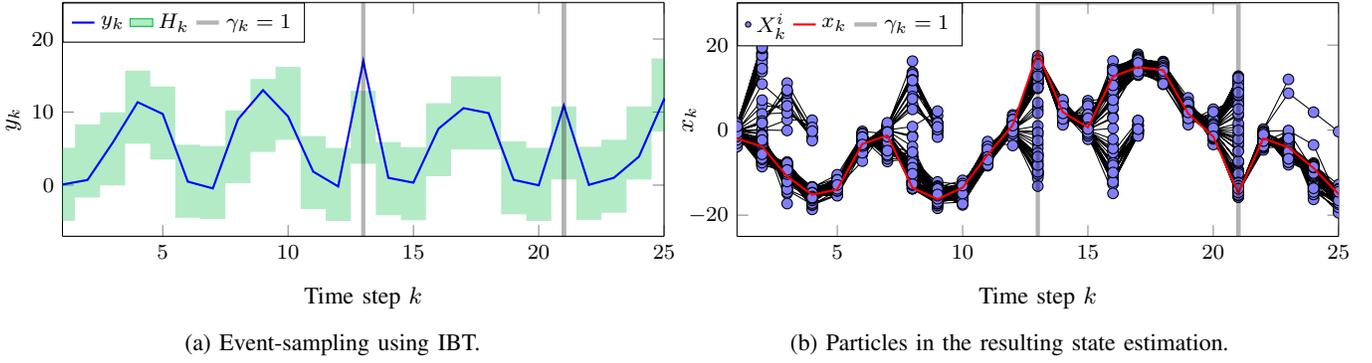

Throughout the paper, the following definitions will be used. Assume $x \in \mathbb{R}^n$ and $y \in \mathbb{R}^m$ where $m,n \in \mathbb{Z}_+$. Let $k_a:k_b := \{ k \in \mathbb{N},~\text{s.t.}~ k_a \leq k \leq k_b \}$. Let $\{a_i\}_{i\in1:b}$ be the set of elements $a_i$ in the interval $i \in 1:b$. Let $\delta_a(x)$ be the Dirac delta function with its mass located at $x = a$. Let $g(x) \propto h(x)$ denote that two functions $g,h$ are proportional, i.e., equal up to a constant factor, and $g(x) \approxprop h(x)$ that they are approximately proportional. Further, let $p(a \mid b)$ denote the density of $a$ conditioned on $b$, and let the density $p(x_k\mid x_{k-1})$ be referred to as the \emph{transition density}, $p(y_k\mid x_k)$ the \emph{likelihood density} and both $p(y_k \mid  x_{k-1})$ or $p(y_k\mid y_{k-1})$ as the \emph{predictive likelihood density} but conditioned on different variables. Let $X \sim p(x)$ denote that the value $X$ is drawn from the density $p(x)$. Finally, let $\mathcal{U}\left(\mathcal{A}\right)$ be the uniform distribution over some set $\mathcal{A} \subset \mathcal{R}^m$, $\mathcal{N}\left(x \mid \mu, \Sigma \right)$ be the multivariate Gaussian distribution with mean $\mu$ and covariance matrix $\Sigma$, and $\mathcal{C}\left(\mathcal{W}\right)$ be the categorical distribution with the corresponding finite set of probabilities $\mathcal{W}$ for the indices.

In state estimation, the goal is to estimate the unknown states $x_{1:k}$ given the measurements $y_{1:k}$. We assume a standard nonlinear state-space model in discrete time,
\begin{align} \label{eq:ssm}
	\begin{split}
		x_{k+1} &= f_k(x_k, w_k) \\
		y_k &= h_k(x_k, v_k)
	\end{split} \qquad \Leftrightarrow
	   \begin{split}
	       x_k &\sim p(x_k\mid x_{k-1}) \\
	       y_k &\sim p(y_k \mid  x_k)
	   \end{split}
\end{align}
where $w_k$ and $v_k$ are noise processes with known distributions. This can be expressed using the transition and likelihood densities as shown, which we consider to be known. Finding the unknown states conditioned on the measurements is thus the same as finding $p(x_k\mid y_{1:k}),~\forall k \in 1:T$. Using Bayes' theorem, this posterior can be expressed as
\begin{equation}
\label{eq:bayesnormal}
p(x_k\mid y_{1:k}) = \frac{p \left( y_k \mid  x_k \right) p\left(x_k \mid  y_{1:k-1}\right)}{p\left( y_k \mid  y_{1:k-1} \right)}
\end{equation}
where
\begin{align}
p\left(x_k \mid  y_{1:k-1}\right) = \int p \left( x_k \mid  x_{k-1} \right) p \left(x_{k-1} \mid  y_{1:k-1} \right) dx_{k-1}
\end{align}
Assuming a linear--Gaussian (LG) system, $p(x_k\mid y_{1:k})$ has a closed-form solution given by the Kalman filter. For a nonlinear or non-Gaussian system, solving the recursive filtering equation is in general intractable.

\subsection{Event-based state estimation}  \label{sec:background_ebse}

In event-based sampling, the observer and sensors are considered disjoint, and thus a measurement is only considered known when it triggers some predefined trigger rule at the sensor. Let $H_k$ be the set of all possible values of $y_k$ at time instance $k$ that do \emph{not} result in a trigger, and define $\gamma_k=1$ as instances where a measurement is sent, i.e. $y_k \not \in H_k$, and $\gamma_k=0$ when $y_k \in H_k$. Given that the trigger rule is known, $\gamma_k=0$ also carries information in the form of $H_k$. This can be captured by defining a hybrid variable $\mathcal{Y}_k = \{y_k\}$ when $\gamma_k=1$, and $\mathcal{Y}_k = H_k$ when $\gamma_k=0$, for which the likelihood becomes
\begin{align}
	p \left( \mathcal{Y}_k \mid  x_k \right) = \int p\left(\mathcal{Y}_k \mid  y_k \right) p\left( y_k \mid  x_k \right) dy_k
\end{align}
Following \cite{Sijs12}, $p\left(\mathcal{Y}_k \mid  y_k \right)$ can be expressed as
\begin{align}
	 p\left(\mathcal{Y}_k \mid  y_k \right) = \begin{cases}
	 	\delta_{y_k}(x) \qquad &\gamma_k = 1 \\
	 	\mathcal{U}(H_k) \qquad &\gamma_k = 0
	 \end{cases}
\end{align}
where $\mathcal{U}(H_k)$ is the uniform distribution over the set $H_k$. This implies that
\begin{align} \label{eq:event_likelihood}
p \left( \mathcal{Y}_k \mid  x_k \right) \propto \begin{cases}
	p \left( y_k \mid  x_k \right) \qquad &\gamma_k = 1 \\
	\int_{y_k \in H_k} 	p \left( y_k \mid  x_k \right) dy_k \qquad &\gamma_k = 0
\end{cases}
\end{align}
and the recursive filtering equation for event-based state estimation becomes
\begin{align} \label{eq:bayesevent}
	p(x_k\mid \mathcal{Y}_{1:k}) = \frac{p \left( \mathcal{Y}_k \mid  x_k \right) p\left(x_k \mid  \mathcal{Y}_{1:k-1}\right)}{p\left( \mathcal{Y}_k \mid  \mathcal{Y}_{1:k-1} \right)}
\end{align}
Using the definitions from \cite{Shi16}, the triggers SOD and IBT construct $H_k$ as
\begin{align}
	H_k^{\text{SOD}} &:= \left\{y_k \in \mathbb{R}^m,~\text{s.t.}\quad \|F_k(y_{k}^- - y_k) \|_{\infty} \leq \Delta \right\} \\
	H_k^{\text{IBT}} &:= \left\{ y_k \in \mathbb{R}^m, ~\text{s.t.}\quad \|F_k(\hat{y}_k - y_k)\|_{\infty} \leq \Delta \right\}
\end{align}
where $y_{k}^-$ is the previous triggered measurement, $\hat{y}_k = \mathbb{E}\left( p(y_k\mid \mathcal{Y}_{1:k-1}) \right)$ and $F_k$ is some weight matrix for normalization. 

\subsection{Particle filters} \label{sec:background_pf}
One popular choice for dealing with the intractable recursive filtering equation \eqref{eq:bayesnormal} is the particle filter \cite{Doucet08}, which approximates $p(x_k\mid y_{1:k})$ as an empirical distribution of $N$ particles $X_k^i$ with corresponding weights $W_k^i$, i.e., $p(x_k\mid y_k) \approx \sum_{i=1}^N W^i_k \delta_{X^i_k}(x_k)$. Given $\{W^i_{k-1},X^i_{k-1}\}_{i\in1:N}$ approximating $\pd{x_{k-1}}{y_{1:k-1}}$, a particle-based approximation for $p(x_k \mid  y_{1:k})$ can be obtained as follows:
\begin{enumerate}
    \item Resample by drawing indices $a^i \sim \mathcal{C}\left( \left\{ W^j_{k-1}\right\}_{j\in1:N} \right)$ with replacement to combat degeneracy,
    \item Draw $X_{k}^i$ from some proposal density $q\left(x_k \mid y_k,  X^{a^i}_{k-1} \right)$,
    \item Calculate new weights $W^i_k \propto \frac{p\left(y_k \mid  X^i_k\right) p\left(X_k^i \mid  X^{a^i}_{k-1}\right)}{q\left(X^i_k \mid  X^{a^i}_{k-1}, y_k\right)}$,
\end{enumerate}
$\forall i\in 1:N$. By initializing $X_0^i \sim p(x_0)$, it can under some minor assumptions on the state-space model and $q(\cdot)$ be shown that $\lim_{N \rightarrow \infty} \hat{p}_N(x_{1:k} \mid  y_{1:k}) = p(x_{1:k} \mid  y_{1:k})$ \cite{Gustafsson10}. Finding a good proposal can be hard, and a common choice is to simply put $q(x_k | y_k, x_{k-1}) = p(x_k | x_{k-1})$, yielding what is known as the \emph{bootstrap particle filter} (BPF) \cite{Gordon93}. 

A common problem in particle filters is the notion of \emph{particle degeneracy}. Depending on the system, situations might arise where only a few particles after a single propagation step end up in areas corresponding to a high measurement likelihood. This implies that most weight will concentrate at these few particles, yielding a poor representation of the posterior. There exist a myriad of different methods to combat the degeneracy problem, see Chapter~4 in \cite{Doucet08} and \cite{Liu14}. One common and simple approach that will be of importance in Section~\ref{sec:event_sample_degeneracy} is the \emph{auxiliary particle filter} (APF) \cite{Pitt99, Whiteley11}, where resampling is conditioned on the predictive likelihood $p \left( y_k \mid x_{k-1} \right)$ such that the indexes $a^i$ are instead drawn from $\mathcal{C}\left( \left\{ V^j_{k-1}\right\}_{j\in1:N} \right)$ where $V^i_{k-1} \propto W^i_{k-1} p \big( y_k \mid  X^i_{k-1} \big)$ and the propagated particles are weighted as
\begin{align}
\begin{split} \label{eq:APF}
W^i_k \propto \frac{p \big( y_k \mid  X^i_k \big) p \big( X^i_k \mid  X^{a^i}_{k-1} \big) }{p \big( y_k \mid  X^{a^i}_{k-1} \big) q \big(X^i_k \mid  X^{a^i}_{k-1}, y_k \big) }\quad \forall i\in 1:N.
\end{split}
\end{align}
Finding $p \left( y_k \mid x_{k-1} \right) $ is in general difficult, and often an approximation has to be used. However, if it is obtainable and used in conjunction with an optimal choice of the proposal distribution, the APF is referred to as \emph{fully adapted} \cite{Whiteley11}.

\subsection{Particle filters under event-based sampling}

A system under event-based sampling can be seen as a simple switching system, where the likelihood changes according to Eq.~\eqref{eq:event_likelihood}. Applying the particle filter to event-based state estimation is thus quite straightforward when considering a simple open-loop trigger such as SOD and the bootstrap filter \cite{Davar17}. Here particles are simply propagated with the transition density and weighted as
\begin{equation} \label{eq:likelihood_per_particle}
    p(\mathcal{Y}_k | X^i_k) = \int_{y_k \in H_k} p(y_k | X_k^i) dy_k
\end{equation}
For a closed-loop trigger, we would further need to calculate $H_{k}$ based on the posterior approximation at time $k-1$. For the IBT, this would resort to approximating the mean of the predictive likelihood as follows \cite{Ruuskanen20}:
\begin{align} \label{eq:estpredlike}
    \begin{split}
	    \hat{y}_k &= \mathbb{E} \left( p(y_k\mid \mathcal{Y}_{1:k-1}) \right) \\
	    &\approx \mathbb{E} \left(\int p(y_k \mid  x_{k}) \frac{1}{N}\sum_{i=1}^N \delta_{X^i_k} dx_k \right) \\
	    & = \frac{1}{N} \sum_{i=1}^N \mathbb{E} \left( p(y_k \mid  X^i_k) \right), \quad X^i_k \sim p\left(x_k \mid  X^{a^i}_{k-1} \right)
	\end{split}
\end{align}
An example displaying the bootstrap particle filter in conjunction with the IBT shown in Figure \ref{fig:trigexample}.

\section{Event sample degeneracy} \label{sec:event_sample_degeneracy}

The issue of degeneracy is a well-known problem in particle filters, and it is hard to believe that it would be anything different for particle filters under event-based sampling. To start, lets consider the simplest event-based estimation strategy where the no-event information $H_k$ is omitted at each $\gamma_k=0$. The filter then reduces to a pure $n$-step prediction of $x_k$ from the previous event. In the absence of measurements after an event, the particles will then start to disperse to cover an area in the state space determined by the variance of this prediction. For most systems the prediction variance is increasing in $n$ and might even be unbounded, which will lead to a considerable particle degeneracy when an event eventually occurs.

Instead, including the no-event information as shown in Section \ref{sec:background_ebse}, this particle dispersion between events will be heavily constrained due to resampling dependent on the likelihood density \eqref{eq:event_likelihood}. As can be seen, for each particle $X^i_k$ at $\gamma_k=0$ this no-event likelihood becomes the integral of the original likelihood over $H_k$. Thus, given a sequence of no-event time instances, the particles of the posterior approximation will disperse to cover the area of the state space corresponding to a significant likelihood over $H_k$. Then at the next event event, the observed $y_k$ will most likely lie somewhere near the ``edges'' of the $H_k$ set. If $H_k$ is large, this gives a big potential area of plausible measurements, resulting in a large area of plausible locations for $X^i_k$ that could give a significant likelihood. But as $p\left(\mathcal{Y}_k | x_k, \gamma_k=1 \right)$ has a more narrow support in $x_k$ than $p\left( \mathcal{Y}_k | x_k, \gamma_k=0  \right)$, only a small area of those plausible locations for $X^i_k$ will actually yield a significant likelihood once $y_k$ has been obtained. 

This implies that by using a generic particle filter, the sudden loss of support in the switching likelihood function can result in an increase of particle degeneracy at event instances. This is especially true if one considers the bootstrap particle filter, where no information regarding $\mathcal{Y}_k$ is taken into account before propagating the particles. In turn, this could seriously affect the posterior approximation both at the event-sample instances, where accuracy is often highly desired, and the performance of the event-based sampling scheme, as future $H_{k+n}$ for closed-loop triggers are dependent on the posterior approximation at $\gamma_k=1$.  

To reduce this event sample degeneracy, most standard schemes for reducing particle degeneracy should be applicable. The auxiliary particle filter (APF) \cite{Pitt99, Whiteley11} is arguably the most intuitive and straightforward method for this purpose. At the heart of the problem lies the fact that only a fraction of the particles with a significant likelihood on $\mathcal{Y}_{k-1}=H_{k-1}$ would actually have a significant predictive likelihood on $\mathcal{Y}_k=y_k$. Including the predictive likelihood in the resampling scheme would thus particles close to the true underlying $x_{k-1}$ to be retained. This is illustrated in Figure \ref{fig:illus_prop}.

\begin{figure}[t]
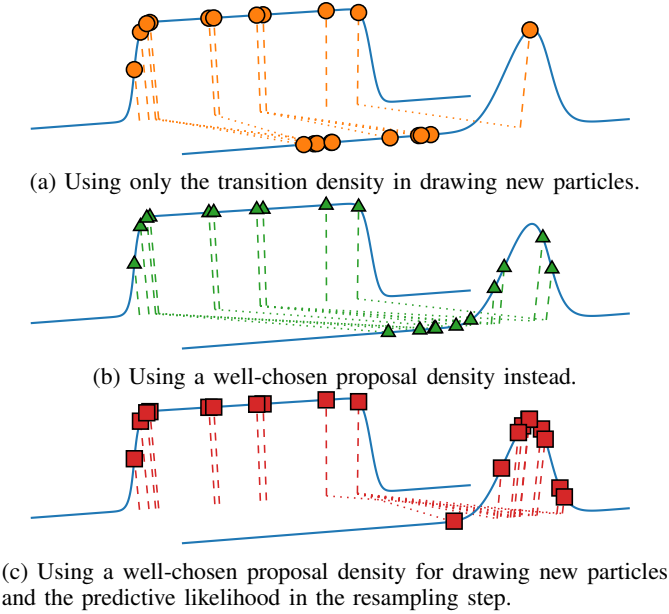

   \centering  
     \begin{subfigure}[t]{\linewidth}
        \centering 
        \setlength\figureheight{3cm}
        \setlength\figurewidth{\linewidth}
    	\input{graphics/show_bpf.tex}
    	\caption{Using only the transition density in drawing new particles.}
	\end{subfigure}
	\begin{subfigure}[t]{\linewidth}
	    \centering 
	    \setlength\figureheight{3cm}
        \setlength\figurewidth{\linewidth}
        \input{graphics/show_sis.tex}
        \caption{Using a well-chosen proposal density instead.}
    \end{subfigure}
    \begin{subfigure}[t]{\linewidth}
        \centering 
        \setlength\figureheight{3cm}
        \setlength\figurewidth{\linewidth}    
        \input{graphics/show_apf.tex}
        \caption{Using a well-chosen proposal density for drawing new particles and the predictive likelihood in the resampling step.}
    \end{subfigure}
  	\medskip
    \caption{Illustration of the effect of using a well-chosen proposal density and a resampling strategy conditioned on the predictive likelihood in a one dimensional event-based setting. The blue lines represent snapshots of the underlying posterior density $p(x_k \mid \mathcal{Y}_{1:k})$, the lines with wide support represents the resulting density for no-event samples ($\gamma_k=0$) while the lines with narrow support represents the density given an event sample ($\gamma_k=1$). The markers represent the particles.}
    \label{fig:illus_prop}
\end{figure}

\subsection{An approximate fully adapted filter for event-based systems} \label{sec:fully_adapted}

An auxiliary particle filter is considered fully adapted if the proposal density is \emph{optimal}, i.e. $q_o\left(x_k \mid y_k, x_{k-1} \right) = \pd{x_k}{y_k, x_{k-1}}$, and if the predictive likelihood $\pd{y_k}{x_{k-1}}$ is known. This is desirable as via Bayes' theorem
$$\pd{x_k}{y_k, x_{k-1}}\pd{y_k}{x_{k-1}} = \pd{y_k}{x_k}\pd{x_k}{x_{k-1}}$$
which gives that the weights \eqref{eq:APF} would become uniform, essentially implying that we are sampling from the unknown posterior density $\pd{x_k}{y_{1:k}}$. However, finding the optimal proposal and predictive likelihood densities are in general difficult, and in practice techniques to approximate these two densities are commonly employed \cite{Whiteley11}. One such common technique is to assume that the joint density $\pd{y_k}{x_k}\pd{x_k}{x_{k-1}} = p(x_k, y_k | x_{k-1})$ can at every time step $k$ be adequately approximated as a multivariate Gaussian,
\begin{align} \label{eq:joint_gauss}
    \begin{split}
    p(x_k, y_k &| x_{k-1}) \approx \mathcal{N} \left(\begin{bmatrix}
        x_k \\
        y_k
    \end{bmatrix} \Bigg|~ \bm{\mu}_k \left(x_{k-1}\right), 
    \bm{\Sigma}_k \left( x_{k-1} \right)
    \right)
    \end{split}
\end{align}
where $\bm{\mu}_k$ and $\bm{\Sigma}_k$ are dependent both on the nonlinear state space model \eqref{eq:ssm} at time step $k$, and the value of $x_{k-1}$. From this, approximations of $q_o\left( x_k \mid, y_k, x_{k-1} \right)$ and $\pd{y_k}{x_{k-1}}$ can be directly obtained. For certain systems, most notably nonlinear systems with additive Gaussian noise, the functions $\bm{\mu}_k\left(x_{k-1}\right)$ and $\bm{\Sigma}_k\left(x_{k-1}\right)$ can be obtained by applying, e.g., linearization \cite{Doucet00} or exact moment matching \cite{Saha09}.

The intractability in obtaining $q_o\left(x_k \mid x_{k-1}, \cY_k \right)$ and $\pd{y_k}{x_{k-1}}$ naturally holds for systems under event-based sampling as well, with the added difficulty of incorporating the new switching likelihood density \eqref{eq:event_likelihood}. But by assuming that \eqref{eq:joint_gauss} holds and approximating the uniform distribution $\mathcal{U}\left(H_k\right)$ in the no-event likelihood density \eqref{eq:event_likelihood} as a Gaussian mixture model, it is possible to derive a general technique for approximating the fully adapted filter in event-based systems. Before continuing, we will need the following two lemmas.
\begin{lemma} \label{lemma:joint_gauss}
Given the joint distribution of two correlated Gaussian distributed variables, $x \in \mathbb{R}^n$ and $y \in \mathbb{R}^m$,
\begin{align}
    p(x, y) = \mathcal{N}
    \left(
    \begin{bmatrix}
        x \\
        y
    \end{bmatrix} \Bigg| 
    \begin{bmatrix}
        \mu_x \\
        \mu_y
    \end{bmatrix}, 
    \begin{bmatrix}
        \Sigma_{xx} & \Sigma_{xy} \\
        \Sigma_{yx} & \Sigma_{yy}
    \end{bmatrix}
    \right),
\end{align}
the conditional distribution becomes 
\begin{align}
	\begin{split}
    p(x \mid y) & = \mathcal{N} \left(x \mid \mu_p, \Sigma_p \right), \\
    \mu_p & = \mu_x + \Sigma_{xy} \Sigma_{yy} ^{-1} \left(y - \mu_y \right), \\
    \Sigma_p & = \Sigma_{xx} - \Sigma_{xy} \Sigma_{yy} ^{-1} \Sigma_{xy}^T.
    \end{split}
\end{align}
\end{lemma}
\smallskip
\begin{proof}
See \cite[p.~116~-~117]{Eaton83}.
\end{proof}
\begin{lemma} \label{lemma:two_gauss}
The product of two independent Gaussian densities becomes a new weighted Gaussian density:
\begin{align}
    \begin{split}
    \mathcal{N} &\left( x \mid \mu_1, \Sigma_1 \right) \mathcal{N} \left(x \mid \mu_2, \Sigma_2 \right) \\
    & = \mathcal{N} \left(\mu_1 \mid \mu_2, \Sigma_1 + \Sigma_2 \right) \mathcal{N} \left(x \mid \mu_3, \Sigma_3\right), \\
    \mu_3& = \left( \Sigma_1^{-1} + \Sigma_2^{-1} \right) \left( \Sigma_1^{-1} \mu_1 + \Sigma_2^{-1} \mu_2 \right), \\
    \Sigma_3& = \left( \Sigma_1^{-1} + \Sigma_2^{-1} \right)^{-1}.
    \end{split}  \raisetag{1em}
\end{align}
\end{lemma}
\smallskip
\begin{proof}
See \cite{Bromiley03}. With a few algebraic manipulations, the expression can be written as two normal densities.
\end{proof}
For brevity, let the dependence on $x_{k-1}$ be implicit and set
$$\begin{bmatrix}
        \mu_x^k \\
        \mu_y^k
    \end{bmatrix} = \bm{\mu}_k\left(x_{k-1}\right), \quad
    \begin{bmatrix}
        \Sigma_{xx}^k & \Sigma_{xy}^k \\
        \Sigma_{yx}^k & \Sigma_{yy}^k
    \end{bmatrix} = \bm{\Sigma}_k \left( x_{k-1} \right)$$
Now, let us consider the two likelihood cases separately.

\textbf{At events ($\bm{\gamma_k=1}$):} Here, $p \left( \mathcal{Y}_k \mid x_k \right)$ reduces to $\pd{y_k}{x_k}$. Considering the joint Gaussian assumption, the approximations of the optimal proposal can be found directly using Lemma \ref{lemma:joint_gauss} and the predictive likelihood by marginalizing over $x_k$ in \eqref{eq:joint_gauss}.
\begin{align} \label{eq:approx_fullyadapted_gamma1}
    \begin{split}
        &q_o\left( x_k | \mathcal{Y}_k, x_{k-1}, \gamma_k=1 \right) \approx \mathcal{N}\left( x_k \mid \mu_p^k, \Sigma_p^k\right) \\
        &p \left( \mathcal{Y}_k | x_{k-1}, \gamma_k=1 \right) \approx \mathcal{N} \left( y_k \mid \mu_y^k, \Sigma_{yy}^k \right) \\
    \end{split}
\end{align} 
where
\begin{align*}
    \begin{split}
        &\quad \mu_p^k = \mu_x^k + \Sigma_{xy}^k \left( \Sigma_{yy}^k \right)^{-1} \left(y_k - \mu_y^k \right) \\
        &\quad \Sigma_p^k = \Sigma_{xx}^k - \Sigma_{xy}^k \left( \Sigma_{yy}^k \right)^{-1} \left(\Sigma_{xy}^k\right)^T
    \end{split}
\end{align*}
\textbf{At no-events ($\bm{\gamma_k=0}$):} Following \cite{Sijs12}, let the uniform distribution over $H_k$ in \eqref{eq:event_likelihood} be approximated with a mixture of $D$ Gaussians with the means located at some set of discretization points $\{z_k^j\}_{j\in 1:D}$ with associated covariance matrices $\{V_k^j\}_{j\in 1:D}$ and weights $\{\alpha_k^j\}_{j\in 1:D}$. The likelihood can thus be approximated as follows:
\begin{align} \label{eq:set_likelihood}
    \begin{split}
        p \left( \mathcal{Y}_k \mid x_k\right) &= \int \mathcal{U} \left( y_k \in H_k \right) p \left( y_k \mid x_k \right)  dy_k \\
        &\approx \sum_{j=1}^D \alpha_k^j  \int \mathcal{N} \big( y_k \mid  z^j_k, V^j_k \big) p \left( y_k \mid x_k \right) dy_k
    \end{split}
\end{align}
If the measurement function has additive Gaussian noise, then $p(y_k|x_k)$ is Gaussian. Otherwise, it can by assumption be approximated as a Gaussian using Lemma \ref{lemma:joint_gauss}. Let the (potentially approximated) Gaussian likelihood be denoted $p(y_k | x_k) = \mathcal{N} \left(y_k \mid \gamma^k_y, \Gamma_y^k \right)$. Then using Lemma \ref{lemma:two_gauss}, the likelihood approximation \eqref{eq:set_likelihood} can be reduced as follows. 
\begin{equation} \label{eq:approx_likelihood}
    \begin{aligned}
       p &( \mathcal{Y}_k \mid x_k ) \approx \sum_{j=1}^D \alpha_k^j  \int \mathcal{N} \big( y_k \mid  z^j_k, V^j_k \big) \mathcal{N} \left( y_k \mid \gamma^k_y, \Gamma^k_y \right) dy_k \\
       &= \sum_{j=1}^D \alpha_k^j \int \mathcal{N} \left( z^j_k \mid \gamma^k_y, V^j_k + \Gamma^k_y \right) \mathcal{N}(y_k \mid \mu_3, \Sigma_3) dy_k \\
        &= \sum_ {j=1}^D \alpha_k^j \mathcal{N} \left( z^j_k \mid \gamma^k_y, V^j_k + \Gamma^k_y \right)
    \end{aligned}
\end{equation}
Using this no-event likelihood approximation, an approximation to the optimal proposal can be found via 
\begin{align} \label{eq:opt_prop_approx}
    \begin{split}
        q_o & \left( x_k  \mid \cY_k, x_{k-1} \right) \propto  \pd{\cY_k}{x_k} \pd{x_k}{x_{k-1}} \\
        &\approx \sum_{j=1}^D \alpha_k^j \mathcal{N} \left( z^j_k \mid \gamma^k_y, V^j_k + \Gamma^k_y \right) \pd{x_k}{x_{k-1}}
    \end{split}
\end{align}
Here, the inner Gaussian $\mathcal{N} \left( z^j_k \mid \gamma^k_y, V^j_k + \Gamma^k_y \right)$ is the same density as (the potentially approximated) $\pd{y_k}{x_k}$ but with an additive variance of $V^j_k$. For each discretization point, we can then via the assumption for \eqref{eq:joint_gauss} create the joint Gaussian
\begin{align} \label{eq:join_gauss_disc}
    \begin{split}
    &\pd{x_k, z^j_k}{x_{k-1}} \\
    &\qquad \approx \mathcal{N} \left(    \begin{bmatrix}
        x_k \\
        z^j_k
    \end{bmatrix} \Bigg|  \begin{bmatrix}
        \mu_x^k \\
        \mu_y^k
    \end{bmatrix}, 
    \begin{bmatrix}
        \Sigma_{xx}^k & \Sigma_{xy}^k \\
        \Sigma_{yx}^k & \Sigma_{yy}^k + V_k^j
    \end{bmatrix}
    \right)
    \end{split}
\end{align}
As $\pd{x_k, z^j_k}{x_{k-1}} = \pd{x_k}{z^j_k, x_{k-1}}\pd{z^j_k}{x_{k-1}}$ using Bayes, an approximate optimal proposal can be derived from \eqref{eq:opt_prop_approx} as the following Gaussian mixture distribution.
\begin{align} \label{eq:approx_optprop_gamma0}
    \begin{aligned}
     &q_o(x_k \mid x_{k-1}, \mathcal{Y}_k, \gamma_k=0) \\
     &\approxprop \sum_{j=1}^D \alpha_k^j p\big(x_k \mid x_{k-1}, z^j_k \big) p \big( z^j_k \mid  x_{k-1} \big) \\
     &= \sum_{j=1}^D \alpha_k^j \mathcal{N}\left( z_k^j \mid \mu_y^k, \Sigma_{yy}^k + V_k^j  \right) \cdot\mathcal{N} \left( x_k \mid \mu_p^{k,j}, \Sigma_p^{k,j} \right)
    \end{aligned}
\end{align} 
where
\begin{align}
    \begin{aligned}
     & \quad  \mu_p^{k,j} = \mu_x^k + \Sigma_{xy}^k \left( \Sigma_{yy}^k + V^j_k \right)^{-1} \left(z^j_k - \mu_y^k \right) \\
     & \quad \Sigma_p^{k,j} = \Sigma_{xx}^k - \Sigma_{xy}^k \left( \Sigma_{yy}^k + V^j_k \right)^{-1} \left(\Sigma_{xy}^k\right)^T
    \end{aligned}
\end{align}
By marginalizing over $x_k$ in \eqref{eq:join_gauss_disc}, an approximation of the predictive likelihood can be obtained as follows.
\begin{equation}
\begin{aligned} \label{eq:approx_predlike_gamma0}
    p(\mathcal{Y}_k \mid& x_{k-1}, \gamma_k=0) \approx \int \sum_{j=1}^D \alpha_k^j  p \big( x_k, z^j_k \mid x_{k-1} \big) dx_k \\
    &\approx \sum_{j=1}^D \alpha_k^j \mathcal{N}\left( z_k^j \mid \mu_y^k, \Sigma_{yy}^k + V_k^j  \right) 
\end{aligned}
\end{equation}

This provides a general and easily implemented approach of approximating the fully-adapted filter for certain systems under event-based sampling, but it is not without its own limitations. First, the performance is dependent on how well $p(x_k, y_k | x_{k-1})$ can be approximated as a Gaussian. The worse the approximation is, the more particles will be needed. It however still holds that $\lim_{N \rightarrow \infty} \hat{p}_N(x_{1:k} \mid  \cY_{1:k}) = p(x_{1:k} \mid  \cY_{1:k})$ if we design the joint Gaussian approximation such that it has a wider support in $x_k, y_k$ than $p(x_k, y_k | x_{k-1})$. Second, how to perform the discretization of the uniform distribution can be discussed. In our experiments, as we are approximating a uniform distribution, we have simply resorted to equidistant gridding of $H_k$, identical covariance matrices $V^j_k = V_k$ and uniform weights $a_k^j = 1/D$ $\forall j \in 1:D, k \in 1:T$. Further, the method suffers from an increased computational burden, as each of the $D$ discretization points will have to be evaluated for each particle giving every iteration for $\gamma_k=0$ a time complexity of $\mathcal{O}(DN)$. Finally, as the number of measurement dimensions $m$ increase, $D$ would have to increase exponentially to retain a similar level of approximation making the approach ill-suited to target systems with large $m$.

\section{Closed-loop triggering in resource-constrained environments} \label{sec:closedloop_trouble}

As event-based sampling reduces the necessary transmissions between the sensor and observer, it has an important application in reducing the sensor energy consumption and network load in resource-constrained remote sensing environments. The fewer transmissions needed, the fewer bytes need to be sent across the network and the longer time the sensor can spend with its radio powered off. 

As the performance of open-loop triggers is inferior to that of closed-loop triggers regarding the error/communication trade-off, being able to fully utilize the benefits of closed-loop triggering will be important for certain applications. However, the inherent disadvantage of closed-loop triggers in that deriving $H_k$ depends on the estimated state of the observer is a problem that needs to be addressed. In event-based literature, this problem is commonly addressed based on one of the two following approaches. 

\subsubsection{Periodic downlink}
The issue is partially ignored by allowing the sensor radio to still receive $H_k$ periodically from the observer, which requires the sensor radio to wake up periodically at sample times \cite{Wu13, Li19}. It is clear that this approach does not make sense for reducing network load or sensor energy consumption in practice unless (i) it is more expensive to transmit than receive for the sensor, and (ii) there does not exist sensor capability to perform adequate local predictions. Otherwise, improvements could be made by either using the local predictor approach or by simply sending $y_k$ at the periodic sensor radio wake-up.

\subsubsection{Local predictor}
When discussing resource-constrained environments in particular, it is common to allow the sensor to run a simple local filter/predictor that can generate predictions of $H_k$, based on the state estimate obtained from the observer at the previous event \cite{Shi14, Kolios2016, Li17, Liu18, SenLi19, Santos2019}. Thus, the sensor radio only needs to wake-up once per event period and only perform one send/receive transmission for the measurement and resulting observer state estimation. This approach, however, necessitates that the simple local predictor is accurate enough to generate an adequate error/communication trade-off, and a more complex sensor and higher energy cost to perform the predictions. In simple linear and Gaussian settings, i.e., the type of settings typically considered in event-based literature, this makes sense as we can often get away with simple predictors that require little computational power. 

However, in situations that require a computationally heavy filter such as the particle filter, this is not something that we can take for granted. The local predictor necessary to achieve an adequate error/communication trade-off might require expensive computational hardware to run, which is, e.g., when considering IoT sensors orthogonal to the desired characteristics of low complexity, small price tag and low in-between sample energy cost \cite{Callebaut2021}. In these cases, an alternative third approach is instead to consider the possibility of precomputing a sequence of $H_{k+1:k+n}$ after each $\gamma_k=1$.

\subsection{Precomputation of trigger bounds}

In event-based state estimation, given $p(x_k \mid \mathcal{Y}_{1:k})$ it is possible to precompute the state estimates at times $k+1:k+n$ by conditioning on $\gamma_{k+1:k+n} = 0$, i.e. that starting from time step $k$ we do not trigger until after $n$ steps. If $\gamma_{k+1:k+n} = 0$, then we know that $y_{k+1:k+n} \in H_{k+1:k+n}$ and thus the posterior $p(x_{k+1:k+n} | \mathcal{Y}_{1:k+n}, \gamma_{k+1:k+n} = 0)$ is independent of the measurements at samples $k+1:k+n$.

Unfortunately, $n$, the number of time steps until the next event, is a stochastic variable that is in general not known in advance. However, by assuming some upper bound $\n$, the observer can at time step $k$ for which $\gamma_k=1$ first estimate $p(x_{k} \mid \mathcal{Y}_{1:k})$ and then sequentially estimate $p(x_{k+i} \mid \gamma_{k+1:k+i-1} = 0, \mathcal{Y}_{1:k})$ $\forall i \in 2:\n$, which can be used to compute and transmit $H_{k+1:k+\n}$ to the sensor at time step $k$. The following approach can then be constructed to remedy the closed-loop issue.

\textbf{The precomputation approach}
\begin{enumerate}
    \item At an event at time step $k$, the sensor radio wakes up, and transmits $y_k$ to observer.
    \item The observer computes $\{p(x_{k+i} | \mathcal{Y}_{1:k}, \gamma_{k+1:k+i-1} = 0)\}_{i\leq \n}$ and $H_{k+1:k+\n}$ for some choice of $\n$.
    \item The sensor receives $H_{k+1:k+\n}$ and powers off its radio.
    \item At each following time step $k+i$, the sensor compares $y_{k+i}$ to $H_{k+i}$ until violated or $i =\n$, in which case a new event is triggered. 
\end{enumerate}
Due to the independence of the measurements, when a triggering eventually occurs at the time step $k+n$, the precomputed posterior estimations $k+1:k+n-1$ is the same that would have been obtained if we computed and transmitted each $H_i \in H_{k+1:k+n}$ after it was known that $\gamma_i=0$. Furthermore, the bound $n \leq \n$ will need to be enforced, such that if time step $k+\n$ is reached, the sensor will force a premature triggering to receive new bounds. This implies that the event period will be shorter than using the periodic downlink approach, however, as the probability to traverse $\n$ samples without triggering decreases with increasing $\n$ it can be seen that $\E{n \mid \n} \rightarrow \E{n}$ as $\n$ grows large.

\subsection{When precomputation is a viable alternative}

Precomputing $H_{k+1:k+\n}$ in the presented manner leads to wasteful computations and transmissions in that all $\n \geq n$ bounds need to be obtained at the sensor in each event period. The wasteful computations, however, lie on the (assumed) less constrained observer and do not burden the sensor except for the waiting time induced from the time the sensor sends $y_k$ until $H_{k+1:k+\n}$ has been remotely computed and received. Furthermore, the transmission of $H_{k+1:k+\n}$ will occur in a single batch at each event. It is well known that it can be more efficient to send more data at once than less data over multiple connections \cite{Bakshi1997}. This is due to that multiple connections require multiple connection initiation overheads which takes time, and costs energy and bandwidth, and that a batch transmission can benefit from different compression strategies. In this case, the data transmitted in $y_k$ and $H_{k}$ is further very small, and thus its cost is comparable to the cost of overhead. 

In all, this makes the precomputation approach a viable alternative to the periodic downlink and local filter approaches in solving the closed-loop issue. It is generally preferable in situations when periodic wake-up and connection initiation cost is too high, and where a cheap local predictor does not give adequate performance for the considered system. Furthermore, compared to the local filter approach, precomputation is a much simpler method to apply in practice, as we do not need to spend resources to construct, implement, nor test the sensor-side local predictor. 

\subsection{Choosing the number of precomputation steps} \label{sec:precomp_choosing_nhat}

Choosing $\n$ for the precomputation approach gives rise to a trade-off between wasted computation and quality. If we trigger before $\n$, then the remaining $\n - n$ steps of trigger bounds and posterior estimations will have been wastefully computed and transmitted. On the other hand, if no triggering occurs in $\n$ samples, the sensor will have to trigger prematurely, reducing the overall performance of the event-based scheme. Considering a computationally cheap filter, the impact of wasted computations could be considered negligible as the observer is assumed to have a high computational capacity, and thus $\n$ can simply be set to some large value. This is harder to argue for if a more computationally expensive method such as the particle filter is used. Furthermore, the distribution of $n$ could very well be dependent on $p(x_k | y_k)$, making it hard to find a static $\n$ sufficient at all $k$. Instead, it would be beneficial to dynamically choose $\n$ in some smart way.

\subsubsection{Estimating the trigger probability}
To choose $\n$ in a smart way, it is vital that we can obtain some estimation of the event triggering probability. First note that, for a new event at time step $k$, we can find the probability to trigger in the next time step as
\begin{align} \label{eq:one_step_trig_prob}
    \begin{split}
        p(\gamma_{k+1} = 1 &\mid \mathcal{Y}_{1:k}) = p(y_{k+1} \not\in H_{k+1} \mid  \mathcal{Y}_{1:k})\\
        &= \int_{y_{k+1} \not\in H_{k+1}} \int_{x_{k+1}} p(y_{k+1} | x_{k+1}) \\
	    & \quad \cdot p(x_{k+1} | \mathcal{Y}_{1:k}) dx_{k+1} dy_{k+1} 
    \end{split}
\end{align}
This integral is in general intractable, but as $p(x_{k+1} | \mathcal{Y}_{1:k})$ is the state propagation distribution, the particle approximation at time step $k$ can be utilized to approximate \eqref{eq:one_step_trig_prob} as
\begin{align} \label{eq:one_step_particle_approx}
	\begin{split}
	p(y_{k+1} &\not\in H_{k+1} \mid  \mathcal{Y}_{1:k}) \\
	&\approx \frac{1}{N} \sum_{i=i}^N \int_{y_{k+1} \not\in H_{k+1}} p(y_{k+1}|X^i_{k+1}) dy_{k+1} \\
	&\qquad X^i_{k+1} \sim p(x_{k+1} | X^{a^i}_{k})
	\end{split}
\end{align}
Denote $p(\gamma^+_{k+i} = 1 \mid \mathcal{Y}_{1:k}) = p(\gamma_{k+i} \mid \gamma_{k+1:k+i-1} = 0, \mathcal{Y}_{1:k})$, i.e. the probability that given no triggering since $k$ we trigger at time step $k+i$. Via the precomputation approach, \eqref{eq:one_step_trig_prob} and \eqref{eq:one_step_particle_approx} can be repeated at each time step $k+i$ to obtain an estimation of $p(\gamma^+_{k+i} = 1 \mid \mathcal{Y}_{1:k})$ sequentially $\forall i \in [1:\n]$. Then using the binary nature of event triggering, the probability to trigger for the first time at $k+n$ since $k$, can be obtained as 
\begin{align} \label{eq:trigger_prob}
    \begin{split}
	    p_T(n\mid \mathcal{Y}_{1:k}) &= p(\gamma_{k+n}^+ = 1\mid \mathcal{Y}_{1:k}) \\
	    & \quad \cdot \prod_{i=1}^{n-1} \big[ 1 - p(\gamma_{k+i}^+ = 1\mid \mathcal{Y}_{1:k}) \big]
    \end{split}
\end{align}
which can be readily estimated.

\subsubsection{Choosing $\n$ by quantile estimation} A simple, ad-hoc dynamic choice of $\widehat{n}$ would then be to base it on the quantile $\alpha$ of $p_T(n)$ via the generalized inverse distribution function
\begin{align} \label{eq:n_quantile}
    n_{\alpha} = \inf \left\{ n\in \mathbb{Z}_+,~\text{s.t.}~  \sum_{i=1}^n p_T(i\mid \mathcal{Y}_{1:k}) \geq \alpha \right\}
\end{align}
which would ensure that the precomputed values would be sufficiently many in a fraction $\alpha$ of cases.

\subsubsection{Choosing $\n$ by optimizing sensor energy consumption} \label{sec:choosing_nhat_opt}

A better choice would arguably be to try to choose $\n$ in a way that would minimize the sensor energy consumption. The choice of $\n$ influences how long the sensor needs to spend with the radio turned on in each event period. Let $\Son$ and $\Soff$ denote the sensor radio in state on and off, respectively. As the cost per time unit for $\Son$ is greater than $\Soff$, minimizing sensor energy consumption equals maximizing the time the sensor spends in $\Soff$.

To do this, we assume that the time spent in $\Son$ is $t_c\n$, where $t_c$ roughly represents the time necessary to compute and transmit 1 sample. The duration of the event period is assumed to be $hn$, where $h$ is the sample time. Here we require that $t_c < h$ but \emph{not} that $\n t_c < h$. Hence, the precomputation duration can span multiple time steps, and the sensor can trigger before all precomputations are completed. In practice, this can be realized by streaming a batch of all currently computed, but not yet sent, values before each time step occurs. The expected time spent in $\Soff$ between two events then becomes 
\begin{align}\label{eq:timespentoff}
    T_o(\n) = \mathbb{E} \left[\max{(hn - t_s\n, 0)}\right]
\end{align}
where the probability mass function of $n$ is dependent on $\n$, and assuming stationarity can be obtained as follows.
\begin{align}
    p\left(n = i | \widehat{n}\right) = \begin{cases}
        p_T(i) \qquad &\text{if} \quad 0 < i < \widehat{n} \\
        1 - \sum_{j=1}^{\widehat{n}-1} p_T(j) \qquad &\text{if} \quad i = \widehat{n} \\
        0  \qquad & \text{otherwise}
    \end{cases}
\end{align}
where $p_T(n) = p_T(n\mid \mathcal{Y}_{1:k})$, and $k$ is the time step of the previous event. 

We can express $T_c(\n) = h^{-1}T_o(\n) = \mathbb{E} \left[\max{(n - c\n, 0)}\right]$ where $c = t_s h^{-1}$, which yields the maximization problem 
\begin{align}\label{eq:argmax_problem}
    \widehat{n}^* = \underset{\n  \in \mathbb{Z}_+}{\text{argmax}} ~T_c(\n)
\end{align}
To identify the maximizers of $T_c(\n)$, we will consider its forward difference $\Delta T_c(\n) = T_c(\n+1) - T_c(\n)$ for which they are obtained $\forall \n$ s.t. $\Delta T_c(\n-1) > 0$ and $\Delta T_c(\n) < 0$.
\begin{lemma} \label{lemma:diff}
The forward difference $\Delta T_c(\n)$ is given by
\begin{align}
\begin{split}
    \Delta T_c(\n) &= \alpha(\n) + 1 - \sum^{\n}_{i=1} p_T(i) - c \left(1 - \sum_{i=1}^{i \leq c\n} p_T(i)\right) \\
    \text{where} & \quad \alpha(\n) = \left(c(\n + 1) - \lfloor c(\n+1) \rfloor \right)p_T\left(\lfloor c(\n+1) \rfloor \right) \\
    &\qquad \qquad \cdot \left(\lfloor c(\n+1) \rfloor - \lfloor c\n \rfloor \right)
\end{split}
\end{align}
where $\lfloor \cdot \rfloor$ is the floor function.
\end{lemma}
\begin{proof}
See the appendix.
\end{proof}
Considering this $\Delta T_c(\n)$, we can conclude that for a general set of triggering probabilities $\{p_T(i) \geq 0\}_{i\in \mathbb{Z}_+}$ s.t. $\sum p_T = 1$, it will be difficult, if not impossible, to find $\n^*$ beyond simply testing all plausible $\n$ via some type of brute force approach. This is not acceptable in a real setting, as by \eqref{eq:one_step_trig_prob} we cannot obtain $p_T(\n)$ before computing the state estimates up to time step $\n-1$. Instead, some heuristics in choosing $\n$ will need to be applied for the general case. One such heuristic is to simply choose the first maximizer that we happen to come across.  For small values of $c$ we could expect the resulting cost to be similar to the global maximizer which is captured in the theorem below.
\begin{theorem}
    The global maximizer of $T_c(\n)$ occurs at some $\n^*$ for which the triggering probability quantile $\alpha > 1 - c$.
\end{theorem}
\begin{proof}
Considering the forward difference in Lemma \ref{lemma:diff}, the value of $c(\n+1) - \lfloor c(\n+1) \rfloor$ in $\alpha(\n)$ is clearly bounded below by $0$. Thus 
\begin{align}
    &\Delta T_c(\n) \geq 1 - \sum^{\n}_{j=1} p_T(j) - c \left(1 - \sum_{j=1}^{j \leq c\n} p_T(j)\right)
\end{align}
Since $c \in (0, 1)$ and $p_T(j)\geq 0~\forall j$, both $\alpha(\n) \geq 0$ and $c\sum_{j=1}^{c\n} p_T(j) \geq 0$, $\Delta T_c(\n)$ can be further bounded below with the following boundary denoted as $\Delta T^S_c(\n)$
\begin{align}
    \Delta T_c(\n) \geq 1 - \sum^{\n}_{j=1} p_T(j) - c = \Delta T_c^S(\n)
\end{align}
As $\Delta T_c^S(\n)$ bounds $\Delta T_c(\n)$ from below for all $\n$, the global maximizer of $T_c(\n)$ can occur first at some 
\begin{align}
    n^*\geq \left(n^S\right)^* &= \inf \left\{ \n^S \in \mathbb{Z}_+,~\text{s.t.}~ 1 - \sum^{\n^S}_{i=1} p_T(i) - c < 0 \right\}
\end{align}
By reordering the condition to $\sum^{\n^S}_{i=1} p_T(i) > 1 - c$, this can be recognize from \eqref{eq:n_quantile} as the generalized inverse distribution function. 
\end{proof}
Without constraining the generality of $\{p_T(i)\}_{i \in \mathbb{Z}_+}$ it becomes difficult to provide performance guarantees on this heuristic. However, all maximizers of $T_c(\n)$ will be located in $\n \in \left( (\n^*)^S, \infty \right)$ where $\sum_{k = (\n^*)^S}^{\infty} p_T(k) < c$. Hence, a small value of $c$ implies that after $\left(n^S\right)^*$, there is not much room for $T_c(\n)$ to dramatically increase. Thus, in these cases, the first maximizer found should give a similar cost as the global maximizer for well-behaved systems.

\section{Computational aspects} \label{sec:comp_aspects}
As state estimation is widely used for real-time control and decision making, it is important to consider the computational performance of the event-based particle filter. Furthermore, if mitigating the closed-loop triggering problem with the precomputation approach, the performance of the heuristic choice of $\n$ as the first local maximizer of $T_c(\n)$ will be improved if the computational time is small. In this section, we will thus discuss the computational aspects of the particle filter for event-based systems. First, the problem of evaluating the no-event likelihood \eqref{eq:event_likelihood} will be addressed. Later, we discuss how the auxiliary particle filter under event-based sampling using the precomputation approach can be efficiently implemented.

\subsection{Evaluating the no-event likelihood density}
For event-based systems, at no event time steps $\gamma_k=0$, the likelihood evaluation when computing the weights for each particle reduces to evaluating the following integral
\begin{align} \label{eq:particle_likelihood}
    p(\mathcal{Y}_k | x_k) \propto \int_{y_k \in H_k} p(y_k|x_k) dy_k
\end{align}
Similarly, if estimating the triggering probability \eqref{eq:one_step_trig_prob}, the following integral needs to be evaluated for each particle
\begin{align} \label{eq:particle_trigg}
    \int_{y_{k+1} \not \in H_{k+1}} p(y_{k+1} |  x_{k+1}) dy_{k+1}
\end{align}
If the dimensionality of the measurements is $m = 1$, this expression can be obtained in closed form if the CDF of $p(y_k|x_k)$ is known. However, if $m > 1$, this integral is in general intractable, which is problematic since it implies that evaluation of the likelihood function must be performed numerically \emph{for each particle}.

In the event-based literature, an approach for tackling the likelihood evaluation has been to use simple Monte Carlo integration \cite{Chitraganti17, Li19},
\begin{align}
    p(\mathcal{Y}_k|X^i_k) \approx \int_{y_k\in H_k} \frac{1}{M}  \sum_j^{M} \delta_{Y^j_k}(y_k) dy_k \quad Y^j_k \sim p(y_k | X^i_k)
\end{align}
for some small $M$. In \cite{Liu18,SenLi19}, this is taken a step further, and $M$ is set to 1 motivated by constraint Bayesian estimation \cite{Shao10}, which essentially reduces the algorithm to acceptance/rejection sampling of particles based on a proposal measurement generated by the particle. This could potentially yield a quick solution to the intractable integrals, but care needs to be taken as the method can suffer from particle depletion. In time steps where the triggering probability is high, the rejection rate of particles will be high and potentially yield a poor posterior approximation. 

We could expect the performance to be closely tied to the size of the event trigger sets $H_{1:T}$, and in turn the communication rate $C_r = \frac{1}{T} \sum_{k=1}^T \gamma_k$ for most well-behaved event-sampled systems. The communication rate is directly related to the triggering probability in each time step as
\begin{align}
    \mathbb{E}(C_r) = \frac{1}{T} \sum_{k=2}^{T}  p(y_{k} \not\in H_{k} \mid  \mathcal{Y}_{1:k-1}) + \frac{1}{T}
\end{align}
assuming that $p(\gamma_1 = 1) = 1$. In each step, we would then expect that a particle in average would have a rejection probability of $p(y_{k} \not\in H_{k} \mid  \mathcal{Y}_{1:k-1})$, and thus on average $1 - p(y_{k} \not\in H_{k} \mid  \mathcal{Y}_{1:k-1})$ of particles will be kept in each step. Thus, in total, the average particle count will be
\begin{align}
    \mathbb{E}(N_r) = \frac{N}{T} \sum_{k=1}^{T-1} \left( 1 - p(y_{k+1} \not\in H_{k+1} \mid  \mathcal{Y}_{1:k} )\right) + \frac{N}{T}
\end{align}
and we can see that, with a low $C_r$, $N_r$ will be close to $N$. A problem occurs in that $p(y_{k} \not\in H_{k} \mid  \mathcal{Y}_{1:k-1})$ might have values close to 1 for certain $k$ even though $C_r$ is low, making it possible for individual time steps to still get a bad posterior approximation. How this distribution of triggering probabilities behaves is specific to the system and event-trigger. But for well-behaved systems with a fitting choice of event-trigger, where a roughly equal part of $p(y_k | \mathcal{Y}_{1:k-1})$ is covered with $H_k$ in each time step, the triggering probabilities would be expected to be situated around the mean. 

In such cases, rejection sampling or using a small $M > 1$, could provide both an adequate performance and a cheap likelihood computation when $m > 1$. If not, the likelihood will likely become a major computational bottleneck. 

\subsection{Implementing the event-based particle filter} \label{sec:implementation}

In its essence, the observer-side implementation of a particle filter under event-based sampling will be very similar to that of a particle filter under periodic sampling. Only a single step that determines $\cY_k$ for the switched likelihood will need to be introduced in the particle propagation. Considering the auxiliary particle filter and assuming that $\hat{p}_N\left(x_{k-1} \mid \cY_{1:k-1} \right) = \sum_{i=1}^N W_{k-1} \delta_{X^i_{k-1}} \left(x_{k-1}\right)$ is given, the posterior approximation $\hat{p}_N\left(x_{k} \mid \cY_{1:k} \right)$ is generated as follows.
\begin{enumerate}
    \item if $\gamma_k=1$ then $\cY_k = y_k$, else generate $\cY_k=H_k$ based on the trigger rule,
    \item Draw resample indices $a^i \sim \mathcal{C}\left( \left\{ V^j_{k-1}\right\}_{j\in1:N} \right)$ where $V^i_k \propto \hat{p}\left(\mathcal{Y}_k \mid X^i_{k-1} \right) W^i_{k-1}$,
    \item Propagate particles as $X^i_{k} \sim q\left(x_k \mid \cY_k, X^{a_i}_{k-1}\right)$,
    \item Weight particles as $W^i_k \propto \frac{p \big( \cY_k \mid  X^i_k \big) p \big( X^i_k \mid  X^{a^i}_{k-1} \big) }{\hat{p} \big( \cY_k \mid  X^{a^i}_{k-1} \big) q \big(X^i_k \mid \cY_k, X^{a^i}_{k-1} \big)}$,
\end{enumerate}
$\forall i \in 1:N$. Setting $\hat{p} \big( \cY_k \mid  x_{k-1} \big) = 1$ returns the nominal particle filter, and further with $q \big(x_k \mid \cY_k, x_{k-1} \big) = \pd{x_k}{x_{k-1}}$ the bootstrap filter. For an open-loop trigger, $H_k$ is obtained solely based on the previous received measurement, and thus step 1) is quick to compute. Instead, using a closed-loop trigger, $H_k$ needs to be computed in some manner. For the local predictor approach, the observer will have to obtain $H_k$ local to the sensor at each $\gamma_k=0$. This is commonly solved by running an identical predictor at the observer.

Using the precomputation or periodic downlink approach, $H_k$ is instead computed from some measure on the one-step prediction of the state distribution $\pd{x_k}{\cY_{1:k-1}}$, e.g. as shown in \eqref{eq:estpredlike} for IBT. Using the BPF, an estimate of $\pd{x_k}{\cY_{1:k-1}}$ can be obtained directly from the particle propagation. If instead the APF is considered, some other way estimating $\pd{x_k}{\cY_{1:k-1}}$ is needed. One such general way is to simply copy the BPF particle propagation, by introducing a second set of particles $\overline{X}^i_k \sim \pd{x_k}{X^{\overline{a}_i}_{k-1}}$ where $\overline{a}^i = \mathcal{C}\left( \left\{ W^j_{k-1}\right\}_{j\in1:N} \right)$ $\forall i \in 1:\overline{N}$.

In the precomputation approach, we at each $\gamma_k = 1$ further need to iterate steps 1)~-~4) assuming $\gamma_{k+1:k+\n} = 0$ until $H_{k+1:k+\n}$ has been computed. If $\n$ is dynamically chosen as presented in Section \ref{sec:precomp_choosing_nhat}, the probability $p_T \left( i \mid \cY_{1:k} \right)$ $\forall i \in 1:\n$ needs to be obtained, but as seen in \eqref{eq:trigger_prob} this can be done \emph{iteratively}. Further, the one step trigger probability estimation \eqref{eq:one_step_particle_approx} can be rewritten as
\begin{align}
    \begin{split} \label{eq:esttrigprobweigths}
    p&(y_{k+1} \not\in H_{k+1} \mid  \mathcal{Y}_{1:k}) \\
	    &\approx \frac{1}{N} \sum_{i=i}^N \left( 1 - \int_{y_{k+1} \in H_{k+1}} p(y_{k+1}|\overline{X}^i_{k+1}) dy_{k+1}\right)
    \end{split}
\end{align}
which allows us for the APF to use our second set of particles $\left\{ \overline{X}^i_{k+1} \right\}_{i\in1:N}$ to estimate the triggering probability via the no-event likelihood. The precomputation approach for generating the set $H_{k+1:k+\n}$ at each $\gamma_k=1$ for the presented approximate fully adapted filter shown in Section \ref{sec:fully_adapted} can be seen in its entirety in Algorithm \ref{alg:iterative}.

\begin{algorithm}[t!]
	\SetAlgoLined
	\SetKwInput{Send}{send}
	\SetKwInput{Generate}{generate}
	\SetKwInput{Set}{set}
    \ForAll{$\gamma_k = 1$, given $\left\{ W^i_{k-1}, X^i_{k-1} \right\}_{i\in1:N}$}{
        $\bm{\mu}_k\left(x_{k-1}\right), \bm{\Sigma}_k\left(x_{k-1}\right) \leftarrow $ e.g. linearization\;
        $q\left(x_k \mid y_k, x_{k-1}\right), \hat{p}\left(y_k \mid x_{k-1}\right) \leftarrow$ \eqref{eq:approx_fullyadapted_gamma1}\;
        $\left\{ W^i_{k}, X^i_{k} \right\}_{i\in1:N} \leftarrow$ \\ \hspace{0.2em} resample/propagate/weight particles\;
    	set $\n=1$, $p_N=1$, $p_T=0$\;
    	\While{$\Delta T_c(\n) \geq 0$ from Lemma \ref{lemma:diff}}{
    	
    	    \tcc{generate trigger bounds and probability}
    	    $\left\{\overline{a}^i \right\}_{i \in 1:N} \leftarrow \mathcal{C}\left( \left\{ W^j_{k-1}\right\}_{j\in1:N} \right)$\;
    	    $\left\{\overline{X}^i_{k+\n} \right\}_{i \in 1:N} \leftarrow \pd{x_k}{X^{\overline{a}_i}_{k+\n-1}}$\;
    	    $H_{k+\n} \leftarrow \textbf{TriggerRule}\left(\left\{\overline{X}^i_{k+\n} \right\}_{i \in 1:N}\right)$\;
    	    $\hat{p}(\gamma_{k+\n} = 1) \leftarrow$ \eqref{eq:esttrigprobweigths} and $\left\{\overline{X}^i_{k+\n} \right\}_{i \in 1:N}$\;
    	    $\hat{p}_T := \hat{p}(\gamma_{k+\n} = 1) \cdot p_N$\;
    		$p_N := p_N \cdot (1 - \hat{p}(\gamma_{k+\n} = 1))$\; 
    	
    	    \vspace{1em}
    	    \tcc{generate posterior approx. assuming no trig}
    	    $\bm{\mu}_{k+\n}\left(x_{k+\n-1}\right), \bm{\Sigma}_{k+\n}\left(x_{k+\n-1}\right) \leftarrow $ e.g. linearization\;
    	    $q\left(x_{k+\n} \mid \cY_{k+\n}, x_{k+\n-1}\right) \leftarrow$ \eqref{eq:approx_optprop_gamma0}\;
    	    $\hat{p}\left(\cY_{k+\n} \mid x_{k+\n-1} \right) \leftarrow$ \eqref{eq:approx_predlike_gamma0}\;
    	    $\left\{ W^i_{k+\n}, X^i_{k+\n} \right\}_{i\in1:N} \leftarrow$ resample/propagate/weight particles\;
    	    \vspace{1em}
    		$\n := \n + 1 $\;
    	}
    	\Send{$H_{k:k+\n}$ to sensor}
    }
	\caption{Implementation of the precomputation approach for some closed-loop triggering rule and the approximate fully adapted filter.}
	\label{alg:iterative}
\end{algorithm}

\section{Evaluation} \label{sec:evaluation}

\begin{figure*}[!t]
	\centering
	\begin{minipage}{\linewidth}
	    \centering
    	\begin{subfigure}[t]{0.33\linewidth}
     		\centering
     		\begin{tikzpicture}[trim axis left]

\begin{axis}[%
width=0.8\columnwidth,
height=0.60\columnwidth,
scale only axis,
separate axis lines,
trim axis left,
yticklabel style={
        /pgf/number format/fixed,
        /pgf/number format/precision=5
},
scaled y ticks=false,
xtick={0, 0.25, 0.5, 0.75, 1},
grid style={dashed,black!20},
grid=major,
legend style={draw, fill=white, font=\footnotesize},
legend pos = north east,
legend columns=1,
legend image post style={scale=2.0},
xlabel={Communication rate $C_r$},
ylabel={Cross-entropy $\varepsilon_{ce}$},
ylabel shift = -0.25 em
]

\addplot [color=blue,only marks, thick, mark=*, mark options={scale=0.5, opacity=0.5}] table [x=X_bpf, y=Y_bpf, col sep=comma] {graphics/data/comparefiltertype_all.csv};
\addlegendentry{BPF (i)};
\addplot [color=blue, only marks, thick, mark=x, mark options={scale=1}] table [x=X_bpf_MC, y=Y_bpf_MC, col sep=comma] {graphics/data/comparefiltertype_all.csv};
\addlegendentry{BPF (ii)};

\addplot [color=red, only marks, thick, mark=*, mark options={scale=0.5, opacity=0.5}] table [x=X_apf, y=Y_apf, col sep=comma] {graphics/data/comparefiltertype_all.csv};
\addlegendentry{APF (i)};
\addplot [color=red, only marks, thick, mark=x, mark options={scale=1}] table [x=X_apf_MC, y=Y_apf_MC, col sep=comma] {graphics/data/comparefiltertype_all.csv};
\addlegendentry{APF (ii)};



\end{axis}

\end{tikzpicture}%
     		\caption{All time steps.}
     		\label{fig:comparefiltertype_all}
    	\end{subfigure}%
    	\begin{subfigure}[t]{0.33\linewidth}
            \centering
     		\begin{tikzpicture}[trim axis left]

\begin{axis}[%
width=0.8\columnwidth,
height=0.60\columnwidth,
scale only axis,
separate axis lines,
trim axis left,
yticklabel style={
        /pgf/number format/fixed,
        /pgf/number format/precision=5
},
scaled y ticks=false,
xtick={0, 0.25, 0.5, 0.75, 1},
grid style={dashed,black!20},
grid=major,
xlabel={Communication rate $C_r$},
]

\addplot [color=blue, only marks, thick, mark=*, mark options={scale=0.5, opacity=0.5}] table [x=X_bpf, y=Y_bpf, col sep=comma] {graphics/data/comparefiltertype_trig.csv};
\addplot [color=blue,only marks, thick, mark=x, mark options={scale=1}] table [x=X_bpf_MC, y=Y_bpf_MC, col sep=comma] {graphics/data/comparefiltertype_trig.csv};

\addplot [color=red, only marks, thick, mark=*, mark options={scale=0.5, opacity=0.5}] table [x=X_apf, y=Y_apf, col sep=comma] {graphics/data/comparefiltertype_trig.csv};
\addplot [color=red, only marks, thick, mark=x, mark options={scale=1}] table [x=X_apf_MC, y=Y_apf_MC, col sep=comma] {graphics/data/comparefiltertype_trig.csv};



\end{axis}

\end{tikzpicture}%
     		\caption{Only time steps where $\gamma_k=1$.}
     		\label{fig:comparefiltertypes_trig}
    	\end{subfigure}%
		\begin{subfigure}[t]{0.33\linewidth}
     	    \centering
     		\begin{tikzpicture}[trim axis left]

\begin{axis}[%
width=0.8\columnwidth,
height=0.60\columnwidth,
scale only axis,
separate axis lines,
trim axis left,
yticklabel style={
        /pgf/number format/fixed,
        /pgf/number format/precision=5
},
scaled y ticks=false,
xtick={0, 0.25, 0.5, 0.75, 1},
grid style={dashed,black!20},
grid=major,
xlabel={Communication rate $C_r$},
]

\addplot [color=blue, only marks, thick, mark=*, mark options={scale=0.5, opacity=0.5}] table [x=X_bpf, y=Y_bpf, col sep=comma] {graphics/data/comparefiltertype_noTrig.csv};

\addplot [color=blue,only marks, thick, mark=x, mark options={scale=1, opacity=1}] table [x=X_bpf_MC, y=Y_bpf_MC, col sep=comma] {graphics/data/comparefiltertype_noTrig.csv};

\addplot [color=red, only marks, thick, mark=*, mark options={scale=0.5, opacity=0.5}] table [x=X_apf, y=Y_apf, col sep=comma] {graphics/data/comparefiltertype_noTrig.csv};

\addplot [color=red, only marks, thick, mark=x, mark options={scale=1, opacity=1}] table [x=X_apf_MC, y=Y_apf_MC, col sep=comma] {graphics/data/comparefiltertype_noTrig.csv};







\end{axis}

\end{tikzpicture}%
     		\caption{Only time steps where $\gamma_k=0$.}
     		\label{fig:comparefiltertypes_noTrig}
    	\end{subfigure}
    	\smallskip
    	\caption{Cross-entropy at different $\Delta$ values for BPF and APF, using either (i) analytic likelihood evaluation or (ii) rejection sampling at $N=100$.}
    	\label{fig:CE_comparsion_D}
	\end{minipage}
	\begin{minipage}{\linewidth}
	    \centering
    	\begin{subfigure}[t]{0.33\linewidth}
        	\centering
		    \begin{tikzpicture}[trim axis left]

\begin{axis}[%
width=0.8\columnwidth,
height=0.60\columnwidth,
scale only axis,
separate axis lines,
trim axis left,
yticklabel style={
        /pgf/number format/fixed,
        /pgf/number format/precision=5
},
scaled y ticks=false,
xtick={0, 100, 200, 300, 400},
grid style={dashed,black!20},
grid=major,
legend style={draw, fill=white, font=\footnotesize},
legend pos = north east,
legend columns=1,
legend image post style={scale=1.0},
xlabel={Number of particles $N$},
ylabel={Cross-entropy $\varepsilon_{ce}$},
ylabel shift = -0.25 em
]

\addplot [color=blue, thick, mark=*, mark options={solid, scale=1.5, draw=black}] table [x=X, y=Y_bpf, col sep=comma] {graphics/data/comparefiltertype_N_all.csv};
\addlegendentry{BPF (i)};
\addplot [color=blue, dashed, thick, mark=x, mark options={solid, scale=1.5}] table [x=X, y=Y_bpf_MC, col sep=comma] {graphics/data/comparefiltertype_N_all.csv};
\addlegendentry{BPF (ii)};

\addplot [color=red, thick, mark=*, mark options={solid, scale=1.5, draw=black}] table [x=X, y=Y_apf, col sep=comma] {graphics/data/comparefiltertype_N_all.csv};
\addlegendentry{APF (i)};
\addplot [color=red, dashed, thick, mark=x, mark options={solid, scale=1.5}] table [x=X, y=Y_apf_MC, col sep=comma] {graphics/data/comparefiltertype_N_all.csv};
\addlegendentry{APF (ii)};



\end{axis}

\end{tikzpicture}%
		    \caption{All time steps.}
		    \label{fig:comparefiltertypes_N_all}
    	\end{subfigure}%
    	\begin{subfigure}[t]{0.33\linewidth}
            \centering
            \begin{tikzpicture}[trim axis left]

\begin{axis}[%
width=0.8\columnwidth,
height=0.60\columnwidth,
scale only axis,
separate axis lines,
trim axis left,
ytick={1, 1.5, 2, 2.5},
yticklabel style={
        /pgf/number format/fixed,
        /pgf/number format/precision=5
},
scaled y ticks=false,
xtick={0, 100, 200, 300, 400},
grid style={dashed,black!20},
grid=major,
xlabel={Number of particles $N$},
ylabel shift = -0.25 em
]

\addplot [color=blue, thick, mark=*, mark options={solid, scale=1.5, draw=black}] table [x=X, y=Y_bpf, col sep=comma] {graphics/data/comparefiltertype_N_trig.csv};
\addplot [color=blue, dashed, thick, mark=x, mark options={solid, scale=1.5}] table [x=X, y=Y_bpf_MC, col sep=comma] {graphics/data/comparefiltertype_N_trig.csv};

\addplot [color=red, thick, mark=*, mark options={solid, scale=1.5, draw=black}] table [x=X, y=Y_apf, col sep=comma] {graphics/data/comparefiltertype_N_trig.csv};
\addplot [color=red, dashed, thick, mark=x, mark options={solid, scale=1.5}] table [x=X, y=Y_apf_MC, col sep=comma] {graphics/data/comparefiltertype_N_trig.csv};



\end{axis}

\end{tikzpicture}%
            \caption{Only time steps where $\gamma_k=1$.}
            \label{fig:comparefiltertypes_N_trig}
    	\end{subfigure}%
		\begin{subfigure}[t]{0.33\linewidth}
      	    \centering
      		\begin{tikzpicture}[trim axis left]

\begin{axis}[%
width=0.8\columnwidth,
height=0.60\columnwidth,
scale only axis,
separate axis lines,
trim axis left,
ytick={1.6, 1.65, 1.7, 1.75},
yticklabel style={
        /pgf/number format/fixed,
        /pgf/number format/precision=5
},
scaled y ticks=false,
xtick={0, 100, 200, 300, 400},
grid style={dashed,black!20},
grid=major,
xlabel={Number of particles $N$},
ylabel shift = -0.25 em
]

\addplot [color=blue, thick, mark=*, mark options={solid, scale=1.5, draw=black}] table [x=X, y=Y_bpf, col sep=comma] {graphics/data/comparefiltertype_N_noTrig.csv};
\addplot [color=blue, dashed, thick, mark=x, mark options={solid, scale=1.5}] table [x=X, y=Y_bpf_MC, col sep=comma] {graphics/data/comparefiltertype_N_noTrig.csv};

\addplot [color=red, thick, mark=*, mark options={solid, scale=1.5, draw=black}] table [x=X, y=Y_apf, col sep=comma] {graphics/data/comparefiltertype_N_noTrig.csv};
\addplot [color=red, dashed, thick, mark=x, mark options={solid, scale=1.5}] table [x=X, y=Y_apf_MC, col sep=comma] {graphics/data/comparefiltertype_N_noTrig.csv};



\end{axis}

\end{tikzpicture}%
      		\caption{Only time steps where $\gamma_k=0$.}
      		\label{fig:comparefiltertypes_N_noTrig}
    	\end{subfigure}
    	\smallskip
    	\caption{Cross-entropy at different $N$ for BPF and APF, using either (i) analytic likelihood evaluation or (ii) rejection sampling for $\Delta = 2.5$, equating a communication rate of $C_r \approx 0.25$.}
    	\label{fig:CE_comparsion_N}
	\end{minipage}
	\medskip\par
	\noindent\makebox[\linewidth]{\rule{\linewidth}{0.4pt}}\par\medskip
\end{figure*}

\begin{figure*}[!t]
	\centering
    \begin{minipage}{\linewidth}
	\centering
	    \begin{subfigure}[t]{0.45\linewidth}
    	    \centering
    	    \begin{tikzpicture}[trim axis left]

\begin{axis}[%
width=0.8\columnwidth,
height=0.40\columnwidth,
scale only axis,
separate axis lines,
trim axis left,
ytick={0, 0.05, 0.1, 0.15},
yticklabel style={
        /pgf/number format/fixed,
        /pgf/number format/precision=5
},
ylabel={Trigger probability $p_T(n)$},
xlabel={Time step $n$},
grid style={dashed,black!20},
grid=major,
legend style={draw, fill=white, font=\footnotesize},
legend pos = north east,
legend columns=1,
legend image post style={scale=1.0}
]

\addplot [color=blue, thick, const plot mark left] table [x=X, y=Y_apf, col sep=comma] {graphics/data/triggerest_trigprob.csv};
\addlegendentry{$\hat{p}_{\text{PF}}(n)$};

\addplot [color=red, thick, const plot mark left] table [x=X, y=Y_mc, col sep=comma] {graphics/data/triggerest_trigprob.csv};
\addlegendentry{$\hat{p}_{\text{MC}}(n)$};

\addplot [name path=upper, const plot mark left, draw=none, forget plot] table [x=X, y=Y_apf_U, col sep=comma] {graphics/data/triggerest_trigprob.csv};
\addplot [name path=lower, const plot mark left, draw=none, forget plot] table [x=X, y=Y_apf_L, col sep=comma] {graphics/data/triggerest_trigprob.csv};
\addplot [color=blue, fill opacity=0.25, area legend, const plot mark left] 
    fill between [of = lower and upper]; 

\addplot [name path=upper, const plot mark left, draw=none, forget plot] table [x=X, y=Y_mc_U, col sep=comma] {graphics/data/triggerest_trigprob.csv};
\addplot [name path=lower, const plot mark left, draw=none, forget plot] table [x=X, y=Y_mc_L, col sep=comma] {graphics/data/triggerest_trigprob.csv};
\addplot [color=red, fill opacity=0.25, area legend, const plot mark left] 
    fill between [of = lower and upper]; 

\end{axis}

\end{tikzpicture}%
    	    \caption{Comparsion between the APF and the naive MC approach with $N = 400$ and $\Delta = 7.5$. Shaded area shows 90\% of data.}
    	    \label{fig:est_trigger_prob_t}
        \end{subfigure}%
        \begin{subfigure}[t]{0.45\linewidth}
    	   	\centering
    		\begin{tikzpicture}[trim axis left]

\begin{axis}[%
width=0.8\columnwidth,
height=0.40\columnwidth,
scale only axis,
separate axis lines,
trim axis left,
grid style={dashed,black!20},
grid=major,
legend style={draw, fill=white, font=\footnotesize},
legend pos = north east,
legend columns=1,
legend image post style={scale=1.0},
xlabel={Number of particles $N$},
ylabel={RMSE},
ylabel shift = -0.25 em
]

\addplot [color=blue, thick, mark=*, mark options={solid,scale=1.5, draw=black}] table [x=X, y=Y_bpf_2, col sep=comma] {graphics/data/triggerest_erroverN.csv};
\addlegendentry{BPF $\Delta = 2.5$};
\addplot [color=blue, dashed, thick, mark=square*, mark options={solid,scale=1.5, draw=black}] table [x=X, y=Y_apf_2, col sep=comma] {graphics/data/triggerest_erroverN.csv};
\addlegendentry{APF $\Delta = 2.5$};

\addplot [color=red, thick, mark=*, mark options={solid,scale=1.5, draw=black}] table [x=X, y=Y_bpf_7, col sep=comma] {graphics/data/triggerest_erroverN.csv};
\addlegendentry{BPF $\Delta = 7.5$};
\addplot [color=red, dashed, thick, mark=square*, mark options={solid,scale=1.5, draw=black}] table [x=X, y=Y_apf_7, col sep=comma] {graphics/data/triggerest_erroverN.csv};
\addlegendentry{APF $\Delta = 7.5$};

\end{axis}

\end{tikzpicture}%
    		\caption{RMSE between the PF and the naive MC approach.}
    		\label{fig:est_trigger_prob_N}
        \end{subfigure}
    	\smallskip
    	\caption{Comparison between the trigger probability estimation from either the naive Monte Carlo approach or via the particle filter over a short time series.}
    	\label{fig:est_trigger_prob}
	\end{minipage}
    \begin{minipage}{\linewidth}
	    \centering
        \begin{subfigure}[t]{0.45\linewidth}
            \centering
    		\begin{tikzpicture}[trim axis left]

\begin{axis}[%
width=0.8\columnwidth,
height=0.40\columnwidth,
scale only axis,
separate axis lines,
trim axis left,
ymin=0,
ymax=3,
ytick={0, 0.5, 1.0, 1.5, 2.0, 2.5, 3.0},
grid style={dashed,black!20},
grid=major,
legend style={draw, fill=white, font=\footnotesize},
legend pos = north east,
legend columns=1,
legend image post style={scale=1.0},
xlabel={Time step $n$},
ylabel={$T_c(n)$},
ylabel shift = -0.25 em
]

\addplot [color=blue, thick] table [x=X, y=Tc_1, col sep=comma] {graphics/data/Tc_delta2.csv};
\addlegendentry{$c=0.05$};
\addplot [color=blue, dashed, thick, forget plot] table [x=X, y=Tc_1_apf, col sep=comma] {graphics/data/Tc_delta2.csv};
\addplot [color=blue, only marks, thick, mark=square*, mark options={scale=1.5, draw=black}, forget plot] table [x=X, y=Tc_1_nopt, col sep=comma] {graphics/data/Tc_delta2.csv};
\addplot [color=blue, only marks, thick, mark=*, mark options={scale=1.5, draw=black}, forget plot] table [x=X, y=Tc_1_nsh, col sep=comma] {graphics/data/Tc_delta2.csv};
\addplot [color=blue, only marks, thick, mark=triangle*, mark options={scale=1.5, draw=black}, forget plot] table [x=X, y=Tc_1_nh_apf, col sep=comma] {graphics/data/Tc_delta2.csv};

\addplot [color=red] table [x=X, y=Tc_2, col sep=comma] {graphics/data/Tc_delta2.csv};
\addlegendentry{$c=0.1$};
\addplot [color=red, dashed, thick, forget plot] table [x=X, y=Tc_2_apf, col sep=comma] {graphics/data/Tc_delta2.csv};
\addplot [color=red, only marks, thick, mark=square*, mark options={scale=1.5, draw=black}, forget plot] table [x=X, y=Tc_2_nopt, col sep=comma] {graphics/data/Tc_delta2.csv};
\addplot [color=red, only marks, thick, mark=*, mark options={scale=1.5, draw=black}, forget plot] table [x=X, y=Tc_2_nsh, col sep=comma] {graphics/data/Tc_delta2.csv};
\addplot [color=red, only marks, thick, mark=triangle*, mark options={scale=1.5, draw=black}, forget plot] table [x=X, y=Tc_2_nh_apf, col sep=comma] {graphics/data/Tc_delta2.csv};

\addplot [color=green] table [x=X, y=Tc_3, col sep=comma] {graphics/data/Tc_delta2.csv};
\addlegendentry{$c=0.25$};
\addplot [color=green, dashed, thick, forget plot] table [x=X, y=Tc_3_apf, col sep=comma] {graphics/data/Tc_delta2.csv};
\addplot [color=green, only marks, thick, mark=square*, mark options={scale=1.5, draw=black}, forget plot] table [x=X, y=Tc_3_nopt, col sep=comma] {graphics/data/Tc_delta2.csv};
\addplot [color=green, only marks, thick, mark=*, mark options={scale=1.5, draw=black}, forget plot] table [x=X, y=Tc_3_nsh, col sep=comma] {graphics/data/Tc_delta2.csv};
\addplot [color=green, only marks, thick, mark=triangle*, mark options={scale=1.5, draw=black}, forget plot] table [x=X, y=Tc_3_nh_apf, col sep=comma] {graphics/data/Tc_delta2.csv};

\end{axis}

\end{tikzpicture}%
    		\caption{$\Delta = 2.5$.}
            \label{fig:Tc_deltalow}
        \end{subfigure}%
         \begin{subfigure}[t]{0.45\linewidth}
            \centering
    		\begin{tikzpicture}[trim axis left]

\begin{axis}[%
width=0.8\columnwidth,
height=0.40\columnwidth,
scale only axis,
separate axis lines,
trim axis left,
ymin=0,
ymax=15,
ytick={0, 2.5, 5.0, 7.5, 10.0, 12.5, 15.0},
grid style={dashed,black!20},
grid=major,
xlabel={Time step $n$},
ylabel shift = -0.25 em
]

\addplot [color=blue, thick] table [x=X, y=Tc_1, col sep=comma] {graphics/data/Tc_delta7.csv};
\addplot [color=blue, dashed, thick] table [x=X, y=Tc_1_apf, col sep=comma] {graphics/data/Tc_delta7.csv};
\addplot [color=blue, only marks, thick, mark=square*, mark options={scale=1.5, draw=black}] table [x=X, y=Tc_1_nopt, col sep=comma] {graphics/data/Tc_delta7.csv};
\addplot [color=blue, only marks, thick, mark=*, mark options={scale=1.5, draw=black}] table [x=X, y=Tc_1_nsh, col sep=comma] {graphics/data/Tc_delta7.csv};
\addplot [color=blue, only marks, thick, mark=triangle*, mark options={scale=1.5, draw=black}] table [x=X, y=Tc_1_nh_apf, col sep=comma] {graphics/data/Tc_delta7.csv};

\addplot [color=red] table [x=X, y=Tc_2, col sep=comma] {graphics/data/Tc_delta7.csv};
\addplot [color=red, dashed, thick] table [x=X, y=Tc_2_apf, col sep=comma] {graphics/data/Tc_delta7.csv};
\addplot [color=red, only marks, thick, mark=square*, mark options={scale=1.5, draw=black}] table [x=X, y=Tc_2_nopt, col sep=comma] {graphics/data/Tc_delta7.csv};
\addplot [color=red, only marks, thick, mark=*, mark options={scale=1.5, draw=black}] table [x=X, y=Tc_2_nsh, col sep=comma] {graphics/data/Tc_delta7.csv};
\addplot [color=red, only marks, thick, mark=triangle*, mark options={scale=1.5, draw=black}] table [x=X, y=Tc_2_nh_apf, col sep=comma] {graphics/data/Tc_delta7.csv};

\addplot [color=green] table [x=X, y=Tc_3, col sep=comma] {graphics/data/Tc_delta7.csv};
\addplot [color=green, dashed, thick] table [x=X, y=Tc_3_apf, col sep=comma] {graphics/data/Tc_delta7.csv};
\addplot [color=green, only marks, thick, mark=square*, mark options={scale=1.5, draw=black}] table [x=X, y=Tc_3_nopt, col sep=comma] {graphics/data/Tc_delta7.csv};
\addplot [color=green, only marks, thick, mark=*, mark options={scale=1.5, draw=black}] table [x=X, y=Tc_3_nsh, col sep=comma] {graphics/data/Tc_delta7.csv};
\addplot [color=green, only marks, thick, mark=triangle*, mark options={scale=1.5, draw=black}] table [x=X, y=Tc_3_nh_apf, col sep=comma] {graphics/data/Tc_delta7.csv};

\end{axis}

\end{tikzpicture}%
    		\caption{$\Delta = 7.5$.}
    		\label{fig:Tc_deltahigh}
        \end{subfigure}
    	\smallskip
    	\caption{The value of $T_c(n)$ according to both estimation methods of $p_T(n)$ with $N=400$ for some values of $c$. The solid line shows $T_c(n)$ evaluated with $\hat{p}_{\text{MC}}(n)$, whereas the dashed line uses $\hat{p}_{\text{PF}}(n)$. The circle shows the $1-c$ quantile, the square the true maximizer and the triangle the maximizer according to the heuristic choice of the first $n$ s.t. $\Delta T_c(n) < 0$.}
    	\label{fig:Tc_comparsion}
	\end{minipage}
	\medskip\par
	\noindent\makebox[\linewidth]{\rule{\linewidth}{0.4pt}}\par\medskip
\end{figure*}

To form an intuition about the performance of the presented improvements, we performed two larger simulation studies on the classical multimodal nonlinear system commonly used for benchmarking particle filters \cite{Li19, Chitraganti17, Gustafsson10, Gordon93}:
\begin{align} \label{eq:nonlin_example_system}
    \begin{split}
        x_{k+1} &= \frac{x_{k}}{2} + \frac{25x_{k}}{1 + x_{k}^2} + 8\cos(1.2k) + w_k \\
        y_k &= \frac{x_k^2}{20} + v_k \\
        \text{where}&\quad w_k \sim \mathcal{N}(0, 1),~ v_k \sim \mathcal{N}(0, 0.1) 
    \end{split}
\end{align}
using event-based sampling with the closed-loop IBT as described in Section \ref{sec:background_ebse}. The studies tested both how the auxiliary particle filter compared to the bootstrap particle filter in an event-based system, how well the filters managed to estimate the triggering probabilities, and how the acceptance/rejection strategy for the no-event likelihoods affected the estimation quality. For the APF, the approximate fully adapted filter described in Section~\ref{sec:fully_adapted} with linearization and with a uniformly discretized $H_k$ with $D=3$ and $V^l_k = \Delta / D$ was used. 

The simulation studies where performed in the \texttt{Julia}\footnote{\url{https://julialang.org/}} programming language, and is available here\footnote{\url{https://github.com/JohanRuuskanen/eventbasedparticlefiltering/tree/preprint}}.

\subsection{Evaluating performance of the event-based particle filter} \label{sec:eval_pf_performace}

In the first simulation study, we examined how well the APF performed compared to the standard BPF in terms of the trade-off between error and communication rate, and how it was affected by using rejection sampling for evaluating the no-event likelihoods. As $\pd{x_k}{y_{1:k}}$ of the considered system is multimodal, an estimation of the cross-entropy 
\begin{equation}
    \varepsilon_{ce} \approx -\frac{1}{T}\sum_{k=1}^T \log \left( \hat{p}(x_k | \mathcal{Y}_{1:k}) \right)
\end{equation}
was used for comparison, which measures the average negative log-likelihood of observing the true state given the posterior approximation. The state estimation was performed over a parameter sweep with $\Delta \in [0, 12]$ and $N \in [25, 400]$. For each parameter setting, the estimation was performed over a randomly generated trajectory $x_{1:T}, y_{1:T}$ where $T = 10^5$.

\subsubsection{Results}
The results are shown in Figure~\ref{fig:CE_comparsion_D} and~\ref{fig:CE_comparsion_N}. As can be seen, the APF outperforms the BPF, especially when considering time steps with events ($\gamma_k=1$) as predicted in Section~\ref{sec:event_sample_degeneracy}. It further seems that the rejection sampling method manages to deliver a performance very close to the analytic likelihood evaluation for most communication rates. For higher rates, where the triggering probability is high, particle depletion due to rejection is a likely cause of the deterioration shown in Figure \ref{fig:comparefiltertypes_noTrig}.

One thing to notice is the surprising result of $\varepsilon_{ce}$ for $\gamma_k = 1$, shown in Figure \ref{fig:comparefiltertypes_trig}. As can be seen, when $C_r$ is very low reducing it further gives a steep improvement in $\varepsilon_{ce}$. One plausible explanation is that this is an artifact of the bounded and multimodal state space of the system. A low $C_r$ implies that $H_k$ is large, and thus particles with significant likelihood on $H_k$ will cover a large part of the bounded state space, potentially both modes completely. A trigger thus implies that the true $x_k$ is an outlier, which might be easy to pinpoint with the given $y_k$. Furthermore, as can be seen in Figure \ref{fig:comparefiltertype_all}, the total $\varepsilon_{ce}$ does not decrease as $C_r$ decreases. 

\subsection{Evaluating the trigger probability estimation and its effect on the heuristic choice of maximizer} \label{sec:trig_prob_estimation}

In the second simulation study, we wanted to examine how well we could manage the online estimation of the triggering probability and how it affects $T_c(\n)$, as it is essential in finding a good horizon for the precomputation approach. Both BPF and the APF using rejection sampling for the likelihood evaluation were considered to form the triggering probability estimation $\hat{p}_{\text{PF}}(n)$. As the true $p_T(n)$ is not known, it was estimated for each particle filter by drawing $x_0,X^i_0 \sim p(x_0)~\forall i \in 1:N$, propagating the system/particles and registering the time steps until the first event. This was then repeated $10^5$ times to form the naive Monte Carlo estimate $\hat{p}_{\text{MC}}(n)$. The entire estimation procedure was further repeated 300 times to provide bounds on both $\hat{p}_{\{\text{MC},\text{PF}\}}(n)$, and run over a parameter sweep with $\Delta = \{2.5, 7.5\}$ and $N \in [25, 400]$.

\subsubsection{Results}
The results can be seen in Figure~\ref{fig:est_trigger_prob}, where we have both compared $p_T(n)$ over the time steps for a certain $N$ and $\Delta$ in Figure~\ref{fig:est_trigger_prob_t}, and the RMSE between $\hat{p}_{\text{PF}}(n)$ and $\hat{p}_{\text{MC}}(n)$ over $N$ for both values of $\Delta$ in Figure~\ref{fig:est_trigger_prob_N}. As can be seen, the bounds of the naive Monte Carlo approach are very tight and can thus be used as a substitute for the unknown true value. For the APF, the estimated triggering probabilities have a larger deviation, but the mean lies close to the mean of $\hat{p}_{\text{MC}}(n)$. 

Using the naive Monte-Carlo and APF estimations of $p_T(n)$, we further examined how well the heuristic choice outlined in Section~\ref{sec:choosing_nhat_opt} of finding the maximizer of $T_c(\n)$ with $\hat{p}_{\text{PF}}(n)$  performed compared the true maximizer of $T_c(\n)$ with $\hat{p}_{\text{MC}}(n)$. Using a single $\hat{p}_{\text{PF}}(n)$ and $\hat{p}_{\text{MC}}(n)$, two varieties of $T_c(\n)$ were created. The results can be seen in Figure~\ref{fig:Tc_comparsion}. Here the circle shows the  $1 - c$ quantile, the square the true maximizer of $T_c(\n)$ evaluated with $\hat{p}_{\text{MC}}(n)$, and the triangle the maximizer found by the heuristic method of $T_c(\n)$ evaluated with $\hat{p}_{\text{PF}}(n)$. As can be seen, both estimation methods of $p_T(n)$ provide similar $T_c(\n)$, and the heuristic $\n$ comes close to the true maximizer for all tested values of $c$. 

\section{Discussion} \label{sec:discussion} 

Reducing event sample degeneracy using the APF as shown in the results from Section \ref{sec:eval_pf_performace} seems promising. However, the benefit of using the APF is more pronounced where $\gamma_k = 1$. Hence, for certain well-behaved systems, it could be viable to run a type of hybrid where the APF would only be employed at $\gamma_k=1$, and the BPF at $\gamma_k = 0$. In these cases, the implementation of the precomputation approach could be made very efficient, as both $\pd{y_k}{\cY_{1:k-1}}$ required for $H_k$ and $\pd{\gamma_k = 1}{\cY_{1:k-1}}$ can be obtain directly via the propagated particles and the corresponding weights of the BPF. 

According to the results shown in Section \ref{sec:trig_prob_estimation}, precomputation with the heuristic maximizer seems like a promising solution to the closed-loop triggering problem. Due to the obtained lower bound of the global maximizer of $T_c(\n)$, we can expect the heuristic to generate near-optimal horizons with respect to the sensor energy consumption as long as $c$, the fraction between time to compute/transmit 1 sample and sampling time, is small. So far no useful upper bound has been found, and we suspect that this will be hard if no assumptions on the probabilities $p_T(n)$ are made.

It is comforting that the no-event likelihood evaluation using rejection sampling seems to be able to give an adequate performance. Numerically evaluating the likelihood using more exact methods would incur a heavy computational cost and render particle filtering for event-based state estimation infeasible in many situations. It is however worth keeping in mind that even though the rejection sampling might not introduce particle degeneracy, it might still give a poor approximation of the posterior. In a real setting, evaluation with a couple of choices of small $M$ to identify a suitable trade-off should be tested. 

As this paper focuses on the single sensor case, a clear future direction would be to consider the case of multiple sensors. This would imply that events could be triggered individually at each sensor, and $H_k$ thus described by a union of both potential measurement values in some dimensions, and actual measurement values in other. At a glance, our method for deriving an approximate fully-adapted filter should be applicable in this setting by simply discretizing $H_k$ in the no-event dimensions. Further, for the precomputation scheme the observer would need to broadcast the updated $H_{k+\n}^s$ to all sensors as soon as one sensors would trigger. Concerning the computational aspects the presented results should be directly applicable, albeit the rejection sampling method for computing the likelihood of more importance.

\section{Conclusion} \label{sec:conclusion}

In this paper we have discussed the three issues of event-sample degeneracy, closed-loop triggering and computational load that occur when using particle filtering under event-based sampling, and provided ways that they can be mitigated.

\section*{Acknowledgment}

This work was partially supported by the Wallenberg AI, Autonomous Systems and Software Program (WASP) funded by the Knut and Alice Wallenberg Foundation. The authors are members of the ELLIIT Excellence Center at Lund University.

\ifCLASSOPTIONcaptionsoff
  \newpage
\fi

\bibliographystyle{IEEEtran}
\bibliography{IEEEabrv,sources}

\appendices
\appendix[Proof of Lemma \ref{lemma:diff}] \label{sec:appendix_lemma_diff}

Start by expanding $T_c(\n)$ as
\begin{align}
    T_c(\n) = \mathbb{E}\left( \max(n - c\n, 0) \right) = \sum_{i > c\n}^{\n} \left(i - c\n \right)p\left(i \mid \n \right) 
\end{align}
Then consider
\begin{align}
    \begin{split}
    &\quad T_c(\n+1) =  \sum_{i>c(\n+1)}^{\n+1} \left(i - c(\n+1) \right) p\left(i \mid \n + 1 \right) \\
    &=  \underbrace{\sum_{i > c(\n+1)}^{\n+1} i p\left(i \mid \n + 1 \right)}_{\textbf{I}} - \underbrace{\sum_{i > c(\n+1)}^{\n+1} c(\n+1) p \left( i \mid \n +1 \right)}_{\textbf{II}}
    \end{split}
\end{align}
where,
\begin{align}
\begin{split}
    \textbf{I} &= \left(\n + 1 \right) \left( 1 - \sum_{i=1}^{\n} p_T(i) \right) + \sum^{\n}_{i > c(\n + 1)} ip_T(i) \\
    &= \left(\n + 1 \right) \left( 1 - \sum_{i=1}^{\n} p_T(i) \right) + \sum^{\n}_{i > c\n} ip_T(i) \\
    &\quad - \lfloor c(\n+1)\rfloor p_T \left( \lfloor c(\n+1)\rfloor \right) \left(\lfloor c(\n+1)\rfloor - \lfloor c\n\rfloor \right) \\
    &= \n - \n\sum_{i=1}^{\n-1} p_T(i) + 1 - \sum_{i=1}^{\n} p_T(i) \ + \sum^{\n-1}_{i > c\n} ip_T(i) \\
    &\quad - \lfloor c(\n+1)\rfloor p_T \left( \lfloor c(\n+1)\rfloor \right) \left(\lfloor c(\n+1)\rfloor - \lfloor c\n\rfloor \right) \\
    &= \sum_{i > c\n}^{\n} i p \left(i \mid \n \right) + 1 - \sum_{i}^{\n} p_T(i) \\
    &\quad - \lfloor c(\n+1)\rfloor p_T \left( \lfloor c(\n+1)\rfloor \right) \left(\lfloor c(\n+1)\rfloor - \lfloor c\n\rfloor \right)
\end{split}
\end{align}
and similarly,
\begin{align}
\begin{split}
    \textbf{II} &=  c\left(\n + 1 \right) \left( 1 - \sum_{i=1}^{\n} p_T(i) \right) + \sum^{\n}_{i > c(\n + 1)} c(\n+1)p_T(i) \\
    &=  c\left(\n + 1 \right) \left( 1 - \sum_{i=1}^{\n} p_T(i) \right) + \sum^{\n}_{i > c\n} c(\n+1)p_T(i) \\
    & \quad - c(\n+1) p_T \left( \lfloor c(\n+1)\rfloor \right) \left(\lfloor c(\n+1)\rfloor - \lfloor c\n\rfloor \right) \\
    &=  c\left(\n + 1 \right) - c\left(\n + 1 \right)\sum_{i=1}^{\n-1} p_T(i) + \sum^{\n-1}_{i > c\n} c(\n+1)p_T(i) \\
    & \quad - c(\n+1) p_T \left( \lfloor c(\n+1)\rfloor \right) \left(\lfloor c(\n+1)\rfloor - \lfloor c\n\rfloor \right) \\
    &=  c\n - c\n \sum_{i=1}^{\n-1} p_T(i) + c - c \sum_{i=1}^{\n-1} p_T(i) \\
    & \quad + \sum^{\n-1}_{i > c\n} c\n p_T(i) +  c \sum^{\n-1}_{i > c\n} p_T(i) \\
    & \quad - c(\n+1) p_T \left( \lfloor c(\n+1)\rfloor \right) \left(\lfloor c(\n+1)\rfloor - \lfloor c\n\rfloor \right) \\
    &= \sum_{i > c\n}^{\n} c\n p \left(i \mid \n \right) + c - c\sum_{i=1}^{i \leq c\n} p_T(i) \\
    & \quad - c(\n+1) p_T \left( \lfloor c(\n+1)\rfloor \right) \left(\lfloor c(\n+1)\rfloor - \lfloor c\n\rfloor \right)
\end{split}\end{align}
Thus $\textbf{I} - \textbf{II}$ becomes
\begin{align}
\begin{split}
    T_c&(\n + 1) = \sum_{i > c \n}^{\n} \left( i - c\n \right) p\left(i \mid \n \right) \\
    &+ 1 - \sum_{i=1}^{\n} p_T(i) - c + c\sum^{i \leq c\n}_{i=1} p_T(i) \\
    &+ \lfloor c(\n+1)\rfloor p_T \left( \lfloor c(\n+1)\rfloor \right) \left(\lfloor c(\n+1)\rfloor - \lfloor c\n\rfloor \right) \\
   \Delta T_c&(\n) = 1 - \sum^{\n}_{i=1} p_T(i) - c\left( 1 -  \sum^{i \leq c\n}_{i=1} p_T(i) \right) \\
    &+ \lfloor c(\n+1)\rfloor p_T \left( \lfloor c(\n+1)\rfloor \right) \left(\lfloor c(\n+1)\rfloor - \lfloor c\n\rfloor \right)
\end{split}
\end{align}

\end{document}